\documentclass[11pt,draftclsnofoot,journal,onecolumn]{IEEEtran}
\usepackage{amsmath,amsfonts,amsthm,amssymb}
\usepackage{setspace}
\usepackage{fancyhdr}
\usepackage{extramarks}
\usepackage{chngpage}
\usepackage[usenames,dvipsnames]{color}
\usepackage{ifthen}
\usepackage{courier}
\usepackage{graphicx}
\usepackage{textcomp}
\usepackage{amssymb}
\usepackage{cancel}
\usepackage{mathrsfs}
\usepackage[font={small}]{caption}
\usepackage{subcaption}
\usepackage{array}
\usepackage{relsize}
\usepackage{multirow,bigdelim}
\usepackage{blkarray}
\usepackage[nospace]{cite}

\newcommand\numberthis{\addtocounter{equation}{1}\tag{\theequation}}

\newcommand\xqed[1]{%
  \leavevmode\unskip\penalty9999 \hbox{}\nobreak\hfill
  \quad\hbox{#1}}
\newcommand\demo{\xqed{$\square$}}
\newcommand\trig{\xqed{$\bigtriangleup$}}

\newtheorem{theo}{\textbf{Theorem}}
\newtheorem{lem}{\textbf{Lemma}}
\newtheorem{corl}{\textbf{Corollary}}

\theoremstyle{definition}
\newtheorem{defi}{\textbf{Definition}}
\newtheorem{ex}{\textbf{Example}}
\newtheorem*{remk}{\textbf{Remark}}

\allowdisplaybreaks

\begin{document}

\title{\fontsize{.71cm}{1em}\selectfont Interference Networks with No CSIT: Impact of Topology}

\author{Navid~Naderializadeh and~A.~Salman~Avestimehr$^\dagger$%
\thanks{$^\dagger$N. Naderializadeh and A. S. Avestimehr are with the Department of Electrical Engineering, University of Southern California, Los Angeles, CA 90089, USA (emails: naderial@usc.edu, avestimehr@ee.usc.edu).
This work is in part supported by a gift from Qualcomm Inc., AFOSR Young Investigator Program Award, ONR award N000141310094, NSF Grants CAREER-0953117, CCF-1161720, and funding from Intel-Cisco-Verizon via the VAWN program.

This work has been presented in part at the IEEE International Symposium on Information Theory (ISIT), 2013 \cite{isit}.}}
\maketitle

\begin{abstract}
We consider partially-connected $K$-user interference networks, where the transmitters have no knowledge about the channel gain values, but they are aware of network topology. We introduce several linear algebraic and graph theoretic concepts to derive new topology-based outer bounds and inner bounds on the symmetric degrees-of-freedom (DoF) of these networks. We evaluate our bounds for two classes of networks  to demonstrate their tightness for most networks in these classes, quantify the gain of our inner bounds over benchmark interference management strategies, and illustrate the effect of network topology on these gains.
\end{abstract}



\section{Introduction}
Channel state information (CSI) plays a central role in the design of physical layer interference management strategies for wireless networks. As a result, training-based channel estimation techniques (i.e. transmission of known training symbols or ``pilots'') are commonly used in today's wireless networks to estimate the channel parameters at the receivers and then to propagate the estimates to other nodes in the network via feedback links. However, as wireless networks grow in size and mobility increases, the availability of channel state information at the transmitters (CSIT) becomes a challenging task to accomplish.

Consequently, there has been a growing interest in understanding how the lack of CSI would have an impact on fundamental limits of interference management in wireless networks. In this work, we focus on the case that channel state information at each node is limited to only a coarse knowledge about network topology. In particular, we consider an interference network consisting of $K$ transmitters and $K$ receivers, where each transmitter intends to deliver a message to its corresponding receiver. In order to model propagation path loss and interference topology, the network is considered to be partially connected (in which ``weak channels'' are removed by setting their channel gains to zero), and the network topology is represented by the adjacency matrix of the network connectivity graph. In this work, we assume that all nodes are aware of network topology, i.e. the adjacency matrix of the network graph, but beyond that, the transmitters have no information about the actual values of the channel gains in the network (i.e., no-CSIT beyond knowing the topology). This is partially motivated by the fact the network connectivity often changes at a much slower pace than the channel gains, hence it is plausible to acquire them at the transmitters. In this setting, the goal is to understand the limits at which the knowledge about network topology can be utilized to manage the interference.



This problem has also been considered in some prior works in the literature. In~\cite{localview}, a slow fading scenario is considered, in which the channel gains associated with existing links in the network remain constant. In this setting, authors have used a ``normalized sum-capacity'' metric in order to characterize the largest fraction of the sum-capacity-with-full-CSI that can be obtained when transmitters only know network topology and the gains of some local channels. Via the characterization of normalized sum-capacity for several classes of network topologies and development of two interference management strategies, namely Coded Set Scheduling and Independent Graph Scheduling, which exploit \emph{temporal neutralization} of interference, it has been shown that the knowledge about network topology can be effectively utilized to increase sum-capacity.

In another work~\cite{jafar_old}, a fast fading scenario is considered in which the channel gains of the existing links in the network are considered to be identically and independently distributed over time (with a sufficiently large coherence time) and also across the users. It has been shown that the DoF region of this problem is bounded above by the DoF region of a corresponding wireless index coding problem, and they are equivalent if both problems are restricted to linear solutions and the coherence time of the channels is sufficiently large. A similar approach has been taken in \cite{jafar}, in which a slow fading scenario is considered and the channel gain values are assumed to be sufficiently large to satisfy a minimum signal-to-noise ratio (SNR) at each receiver. It is shown that, quite interestingly, the degrees-of-freedom (DoF) of this problem has a counterpart in the capacity analysis of wired networks. This connection enabled the derivation of several outer bounds on the symmetric DoF, and the development of an interference-alignment-based achievability scheme. Also, necessary and sufficient conditions have been derived for networks to achieve a symmetric DoF of $\frac{1}{2}$.

In this paper, we focus on a fast fading scenario in which the channel gains change at each time instant according to an i.i.d. distribution (i.e., coherence time of 1). We also assume that the channel gain values are not available at the transmitters and they only have access to the topology knowledge of the network. The assumption of coherence time of 1 is an extreme case of prior works \cite{localview,jafar_old,jafar}, in which transmitters are not able to exploit temporal neutralization or temporal alignment of the interference. Hence, in this setting, our goal is to study the achievable degrees-of-freedom without relying on temporal neutralization or alignment of the interference.

To this end, we derive new graph theoretic and linear algebraic inner and outer bounds on the symmetric DoF of the network. To derive the outer bounds, we will introduce two novel linear algebraic concepts, namely ``generators'' and ``fractional generators'', and utilize them to upper bound the symmetric DoF for general network topologies. The key idea of generators is that in any network topology, we seek for a number of signals from which we can decode the messages of all the users, and then we will find an upper bound for the joint entropy of those signals. Instead of upper bounding this joint entropy proportionally to the number of the signals, we will use the concept of fractional generators to find the tightest upper bound on the entropy of the signals based on the interference interactions at the receivers. Through examples, we will demonstrate that we can systematically apply our outer bounds to any arbitrary network topology. These outer bounds are applicable to any channel coherence time.

Moreover, we will present three inner bounds based on three achievable schemes. First, we discuss two benchmark schemes and characterize their achievable symmetric DoF with respect to two graph theoretic parameters of the network graph, namely ``maximum receiver degree'' and ``fractional chromatic number''. Through examples, we show that these schemes are suboptimal in some networks and gain can be accomplished by taking more details of network topology into account. This motivates the third scheme, called ``structured repetition coding'', which performs at least the same as or strictly better than the two benchmark schemes. The main idea of this scheme is to enable neutralization of interference at the receivers by repeating the symbols based on a carefully-chosen structure at the transmitters. We derive graph theoretic conditions, based on the matching number of bipartite graphs induced by network topology and the repetition structure of transmitters, that characterize the symmetric DoF achieved by structured repetition coding. This scheme can also be applied to any channel coherence time by means of interleaving. Thus, the coherence time of 1 is the worst case in this sense.

Finally, we will evaluate our inner and outer bounds in order to characterize the symmetric DoF in two distinct network scenarios. First, we consider 6-user networks composed of 6 square cells in which each receiver may receive interference from at least one and at most three of its adjacent cells, as well as the signal from its own transmitter. Interestingly, after removing isomorphic graphs, we see that our inner and outer bounds meet for all 22336 possible network topologies except 16, hence characterizing the symmetric DoF in those cases. This implies that in most of these networks, temporal alignment of the interference cannot provide any additional DoF-gain over structured repetition coding.

We also consider 6-user networks with 1 central and 5 surrounding base stations and evaluate our inner and outer bounds for a large number of randomly generated client locations. In this case, the results show that our bounds are tight for all generated network topologies, leading to similar conclusions to the previous case. In both the aforementioned scenarios, we will demonstrate the distribution of the gain of structured repetition coding over the two benchmark schemes and study the impact of network density on these gains.

\noindent \textbf{Other Related Works.}
In the context of interference channels, various settings for limited knowledge of channel state information at the transmitters have been studied in the literature, such as no CSIT (see e.g., \cite{zchannel,huang,vaze_no_csit}) and delayed CSIT (see e.g., \cite{sina,ali,maleki,abdoli,vaze}). However, in this work we consider the case where the transmitters have only a coarse knowledge of the channel gains. In fact, we assume that the transmitters do not know the actual channel gain values, but are equipped with one bit of feedback for each channel showing whether or not the channel is strong enough. There have also been several works in the literature that utilize the specific structure of network topology for interference management (see e.g.,  \cite{vvv,wang,gastpar,hong}); however in these works both the channel gains and network topology are assumed to be known at the transmitters.

\vspace*{\baselineskip}

The rest of the paper is organized as follows. In Section \ref{form_not} we introduce the problem model and notations. In Section \ref{converse} we present our outer bounds on the symmetric DoF. In Section \ref{ach} we present our achievable schemes. In Section \ref{sim} we present our numerical analysis for the aforementioned network scenarios. Finally, we conclude the paper in Section \ref{conc}.

\section{Problem Formulation and Notations}\label{form_not}

A $K$-user interference network ($K\in\mathbb{N}$) is defined as a set of $K$ transmitter nodes $\left\lbrace\text{T}_i\right\rbrace_{i=1}^{K}$ and $K$ receiver nodes $\left\lbrace\text{D}_i\right\rbrace_{i=1}^{K}$. To model propagation path loss and interference topology, we consider a similar model to \cite{jafar} in which the network is partially connected represented by the adjacency matrix $\mathbf{M}\in\{0,1\}^{K\times K}$, such that $\mathbf{M}_{ij}=1$ iff transmitter $\text{T}_i$ is connected to receiver $\text{D}_j$ (i.e. $\text{D}_j$ is in the coverage radius of $\text{T}_i$). We assume there exist direct links between each transmitter $\text{T}_i$ and its corresponding receiver $\text{D}_i$ (i.e. $\mathbf{M}_{ii}=1$, $\forall i\in[1:K]$, where we use the notation $[1:m]$ to denote $\{1,2,...,m\}$ for $m\in\mathbb{N}$). We also define the set of interfering nodes to receiver $\text{D}_j$ as $\mathcal{IF}_j:=\{i:\mathbf{M}_{ij}=1,\	i\neq j\}$.

The communication is time-slotted. At each time slot $l$ ($l\in\mathbb{N}$), the transmit signal of transmitter $\text{T}_i$ is denoted by $X_i[l]\in\mathbb{C}$ and the received signal of receiver $\text{D}_j$ is denoted by $Y_j[l]\in\mathbb{C}$ given by
\begin{align*}
Y_j[l]=g_{jj}[l]X_j[l]+\sum_{i\in\mathcal{IF}_j} g_{ij}[l]X_i[l]+Z_j[l],
\end{align*}
where $Z_j[l]\sim\mathcal{CN}(0,1)$ is the additive white Gaussian noise and $g_{ij}[l]$ is the channel gain from transmitter $\text{T}_i$ to receiver $\text{D}_j$ at time slot $l$. If transmitter $\text{T}_i$ is not connected to receiver $\text{D}_j$ (i.e. $\mathbf{M}_{ij}=0$), then $g_{ij}[l]$ is assumed to be identically zero at all times. We assume that the non-zero channel gains (i.e. $g_{ij}[l]$'s s.t. $\mathbf{M}_{ij}=1$) are independent and identically distributed (with a continuous distribution $f_G(g)$) through time and also across the users, and are also independent of the transmit symbols. The distribution $f_G(g)$ needs to satisfy three regularity conditions: $\mathbb{E}[|g|^2]<\infty$, $f_G(g)=f_G(-g),\forall g\in\mathbb{C}$, and $\exists f_{max}\text{ s.t. }f_{|G|}(r)\leq f_{max},\forall r\in\mathbb{R}^+$, where $f_{|G|}(.)$ is the distribution of $|g|$. The noise terms are also assumed i.i.d. among the users and the time slots, and also independent of the transmit symbols and channel gains.

It is assumed that the transmitters $\{\text{T}_i\}_{i=1}^K$ are only aware of the connectivity pattern of the network (or the network topology), represented by the adjacency matrix $\mathbf{M}$, and also the distribution $f_G$ of the non-zero channel gains; i.e. the transmitters only know which users are interfering to each other and they also know the statistics of the channel gains, not the actual gains of the links. In this paper, we refer to this assumption as no channel state information at the transmitters (no CSIT). As for the receivers $\{\text{D}_j\}_{j=1}^K$, we assume that they are aware of the adjacency matrix $\mathbf{M}$ and the channel gain realizations of their incoming links. In other words, receiver $\text{D}_j$ is aware of $\mathbf{M}$ and $g_{ij}[l]$, $\forall i\in\{j\}\cup\mathcal{IF}_j$, $\forall l$.

In this network, every transmitter $\text{T}_i$ intends to deliver a message $W_i$ to its corresponding receiver $\text{D}_i$. The message $W_i$ is encoded to a vector $X_i^n=[X_i[1]\:X_i[2]\:\hdots\:X_i[n]]^T\in\mathbb{C}^n$ through an encoding function $e_i(W_i|\mathbf{M},f_G)$; i.e. transmitters use their knowledge of network topology and the distribution of the channel gains to encode their messages. There is also a transmit power constraint $\mathbb{E}\left[\frac{1}{n}\| X_i^n \|^2\right]\leq P$, $\forall i\in[1:K]$. This encoded vector is transmitted within $n$ time slots through the wireless channel to the receivers. Each receiver $\text{D}_j$ receives the vector $Y_j^n=[Y_j[1]\:Y_j[2]\:\hdots\:Y_j[n]]^T$ and uses a decoding function $e'_j(Y_j^n|\mathbf{M},\mathcal{G}_j^n)$ to recover its desired message $W_j$. Here, $\mathcal{G}_j^n:=\{g_{ij}^n: i\in[1:K]\}$ where $g_{ij}^n:=[g_{ij}[1]\:g_{ij}[2]\:\hdots\:g_{ij}[n]]^T$ denotes the vector of the channel gain realizations from transmitter $\text{T}_i$ to receiver $\text{D}_j$ during $n$ time slots. We also denote the set of all channel gains in all time slots by $\mathcal{G}^n=\{\mathcal{G}_1^n,...,\mathcal{G}_K^n\}$.

The rate of transmission for user $i$ is denoted by $R_i(P):=\frac{\log|W_i(P)|}{n}$ where $|W_i(P)|$ is the size of the message set of user $i$ and we have explicitly shown the dependence of $W_i$ on $P$. Denoting the maximum error probability at the receivers by $\text{Pr}_e(P)=\underset{j\in[1:K]}{\max}\text{Pr}\left[W_j(P)\neq e'_j(Y_j^n|\mathbf{M},\mathcal{G}_j^n)\right]$, a rate tuple $(R_1(P),...,R_K(P))$ is said to be achievable if $\text{Pr}_e(P)$ goes to zero as $n$ goes to infinity.

In this paper, the considered metric is the symmetric degrees-of-freedom (DoF) metric, which is defined as follows. If a rate tuple $(R_1(P),...,R_K(P))$ is achievable and we let $d_i=\underset{P\rightarrow\infty}{\lim}\frac{R_i(P)}{\log(P)}$, then the DoF tuple of $(d_1,...,d_K)$ is said to be achievable. The symmetric degrees-of-freedom $d_{sym}$ is defined as the supremum $d$ such that the DoF tuple $(d,...,d)$ is achievable.

Therefore, the main problem we are going to address in this paper is that given a $K$-user interference network with adjacency matrix $\mathbf{M}$ (which is known by every node in the network) and channel gains distribution $f_G$, what the symmetric degrees-of-freedom $d_{sym}$ is, under no-CSIT assumption. We will start by presenting our outer bounds on $d_{sym}$ in the next section.

\section{Outer Bounds on $d_{sym}$}\label{converse}

In this section, we will present our outer bounds for the symmetric DoF of $K$-user interference networks. To this end, we provide two types of outer bounds and we will motivate each outer bound through an introductory example. The main idea in both of the outer bounds is to create a set of signals by which we can sequentially decode the messages of all the users with a finite number of bits provided by a genie. This set of signals corresponds to a matrix called a \emph{generator}. We will show systematically that for any network topology, there are some linear algebraic conditions that a matrix should satisfy to be called a generator. Therefore, our outer bounds rely highly on the topology of the network graph and the goal is to algebraically explain how these bounds are derived. The first converse generally states that the number of signals corresponding to a generator is an upper bound for the sum degrees-of-freedom of the network. However, the second converse enhances the first one, showing that there may be tighter upper bounds on the sum degrees-of-freedom due to the specific topology of the network.

For all the outer bounds presented in this section, because of the no-CSIT assumption, we will be replacing \emph{statistically similar} signals with each other, i.e., signals which have the same probability distribution functions. This is due to the fact that the decoding error probability only depends on the marginal channel transition probabilities $p(Y|X_1,...,X_K)$. In particular, we will be using the following lemma in developing our outer bounds on the symmetric degrees-of-freedom.

\begin{lem}\label{sep}
The capacity region, and therefore the degrees-of-freedom, of a $K$-user interference network only depend on the marginal transition probabilities of the channels.
\end{lem}

\subsection{Upper Bounds Based on the Concept of Generators}

We start by presenting our first outer bound through the notion of generators. The main idea of this outer bound is presented in Example \ref{ex1}. Before starting the example, we need to define some notation.

\begin{itemize}
\item If $\mathcal{S}\subseteq[1:K]$ is a subset of users in a $K$-user interference network with adjacency matrix $\mathbf{M}$, then $\mathbf{M}^\mathcal{S}$ denotes the adjacency matrix of the corresponding subgraph and $\mathbf{I}^\mathcal{|S|}$ denotes the $|\mathcal{S}|\times|\mathcal{S}|$ identity matrix.

\item For a general $m\times n$ matrix $\mathbf{A}$ and $\mathcal{N}\subseteq [1:n]$, $\mathbf{A}_\mathcal{N}$ denotes the submatrix of $\mathbf{A}$ composed of the columns whose indices are in $\mathcal{N}$. For the sake of brevity, if $\mathcal{N}=\{i\}$, i.e. if $\mathcal{N}$ has only one member, we use $\mathbf{A}_i$ to denote the $i^{th}$ column of $\mathbf{A}$.

\item For a general matrix $\mathbf{A}$, $c(\mathbf{A})$ denotes the number of columns of $\mathbf{A}$.
\end{itemize}

We will also need the following definition.

\begin{defi}\label{d0}
If $\mathbf{v}\in\{0,\pm1\}^{n\times1}$ and $\mathcal{V}$ is a subspace of $\mathbb{R}^n$, then $\mathbf{v}\in^\pm \mathcal{V}$ means that there exists a vector $\tilde{\mathbf{v}}$ in $\mathcal{V}$ which is the same as $\mathbf{v}$ up to the sign of its elements; i.e.,
\begin{equation*}
\mathbf{v}\in^\pm \mathcal{V} \Leftrightarrow \exists \tilde{\mathbf{v}}\in \mathcal{V} \text{ s.t. } |\tilde{\mathbf{v}}_j|=|\mathbf{v}_j|,\:\forall j\in[1:n].
\end{equation*}

Moreover, if $i$ is an index in $[1:n]$, then $\mathbf{v}\in_i^\pm \mathcal{V}$ implies that there exists a vector $\tilde{\mathbf{v}}$ in $\mathcal{V}$ whose $i^{th}$ element is the same as the $i^{th}$ element of $\mathbf{v}$ up to its sign, while every other element of $\tilde{\mathbf{v}}$ either equals zero or matches the corresponding element of $\mathbf{v}$ up to its sign. To be precise, we have the following definition.
\begin{equation*}
\mathbf{v}\in_i^\pm \mathcal{V} \Leftrightarrow \exists \tilde{\mathbf{v}}\in \mathcal{V} \text{ s.t. } |\tilde{\mathbf{v}}_i|=|\mathbf{v}_i|\text{ and }\tilde{\mathbf{v}}_j(|\tilde{\mathbf{v}}_j|-|\mathbf{v}_j|)=0,\:\forall j\in[1:n]\setminus\{i\}.
\end{equation*}

\trig
\end{defi}

\begin{ex}\label{ex1}
Consider the 5-user interference network in Figure \ref{fig1}. We claim that the symmetric DoF of this network with no CSIT is upper bounded by $\frac{2}{5}$.
\begin{figure}[h]
\centering
\includegraphics[trim = 2in 3in 2in 4.1in, clip,width=0.3\textwidth]{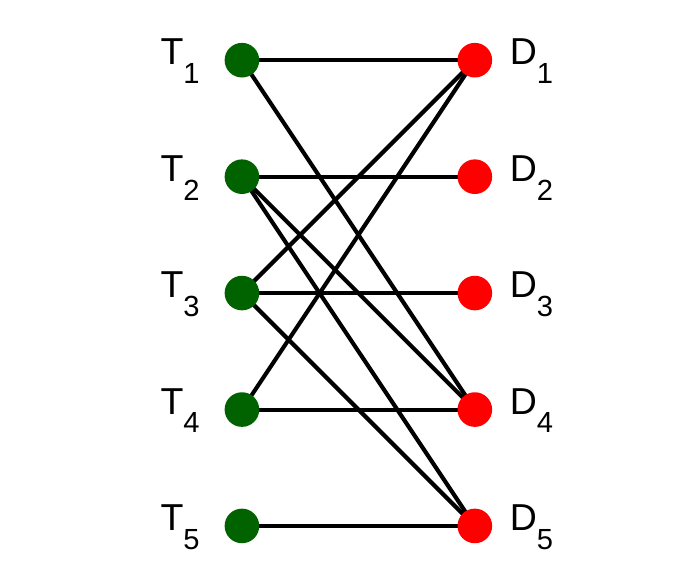}
\caption{A 5-user interference network in which $d_{sym}\leq \frac{2}{5}$.}
\label{fig1}
\end{figure}

Suppose rates $R_i$, $i\in[1:5]$, are achievable. We define the signals
\begin{align*}
\tilde{Y}_1^n&=g_1^nX_1^n+g_3^nX_3^n+g_4^nX_4^n+\tilde{Z}_1^n\\
\tilde{Y}_5^n&=g_2^nX_2^n+g_3^nX_3^n+g_5^nX_5^n+\tilde{Z}_5^n,
\end{align*}
where $\tilde{Z}_1^n$ and $\tilde{Z}_5^n$ have the same distributions as the original noise vectors, but are independent of them and also of each other and $g_i^n=g_{ii}^n$, $i\in[1:5]$. We now show that $H(W_1,...,W_5|\tilde{Y}_1^n,\tilde{Y}_5^n,\mathcal{G}^n)\leq no(\log(P))+n\epsilon_{n}$, which implies
\begin{align*}
\sum_{i=1}^5 R_i&=\frac{1}{n}H(W_1,...,W_5|\mathcal{G}^n)\\
&=\frac{1}{n}\left[I(W_1,...,W_5;\tilde{Y}_1^n,\tilde{Y}_5^n|\mathcal{G}^n)+H(W_1,...,W_5|\tilde{Y}_1^n,\tilde{Y}_5^n,\mathcal{G}^n)\right]\\
&\leq 2\log(P)+o(\log(P))+\epsilon_n,
\end{align*}
hence $d_{sym}\leq \frac{2}{5}$. This is obtained through the following steps, which are explained intuitively here and their formal proof is discussed in the proof of Theorem \ref{th1} for general network topologies.
\begin{itemize}
\item Step 1: $H(W_1,W_5|\tilde{Y}_1^n,\tilde{Y}_5^n,\mathcal{G}^n)\leq n\epsilon_n$, due to the fact that $\tilde{Y}_1^n$ and $\tilde{Y}_5^n$ are statistically the same as $Y_1^n$ and $Y_5^n$, respectively, followed by Lemma \ref{sep} and Fano's inequality.

\item Step 2: $H(W_4|\tilde{Y}_1^n,\tilde{Y}_5^n,W_1,W_5,\mathcal{G}^n)\leq no(\log(P))+n\epsilon_{n}$. This is obtained by noting that from $W_5$, one can create $X_5^n$ and then by using the other terms in the conditioning, we can construct $\tilde{Y}_4^n=\tilde{Y}_1^n-\tilde{Y}_5^n+g_5^n X_5^n=g_1^nX_1^n-g_2^nX_2^n+g_4^nX_4^n+\tilde{Z}_1^n-\tilde{Z}_5^n$, which is statistically the same as $Y_4^n$ except for a larger, but bounded, noise variance, and because of Lemma \ref{sep}, it is able to decode $W_4$. The statistical equivalence follows from the assumption that the distribution of the channel gains is symmetric around zero ($f_G(g)=f_G(-g)$, $\forall g\in\mathbb{C}$). The desired inequality then follows, where the $n\epsilon_n$ term is due to Fano's inequality and the $no(\log(P))$ term is due to the larger noise variance, treated more formally in Lemma \ref{lem2} which appears later.

\item Step 3: $H(W_3|\tilde{Y}_1^n,\tilde{Y}_5^n,W_1,W_5,W_4,\mathcal{G}^n)\leq n\epsilon_n$, obtained by noting that from $W_1$ and $W_4$, one can create $X_1^n$ and $X_4^n$ and then by using the other terms in the conditioning, we can construct $\tilde{Y}_3^n=\tilde{Y}_1^n-g_1^n X_1^n-g_4^n X_4^n=g_3^nX_3^n+\tilde{Z}_1^n$, which is statistically the same as $Y_3^n$. The inequality then follows from Lemma \ref{sep} and Fano's inequality.

\item Step 4: $H(W_2|\tilde{Y}_1^n,\tilde{Y}_5^n,W_1,W_5,W_4,W_3,\mathcal{G}^n)\leq n\epsilon_n$, obtained by noting that from $W_3$ and $W_5$, one can create $X_3^n$ and $X_5^n$ and then by using the other terms in the conditioning, we can construct $\tilde{Y}_2^n=\tilde{Y}_5^n-g_3^n X_3^n-g_5^n X_5^n=g_2^nX_2^n+\tilde{Z}_2^n$, which is statistically the same as $Y_2^n$. The inequality then follows from Lemma \ref{sep} and Fano's inequality.

\end{itemize}

Adding the above inequalities and using the chain rule for entropy yield the desired result. Consequently, starting from $\{\tilde{Y}_1^n,\tilde{Y}_5^n\}$, we created a sequence of users \{1,5,4,3,2\} in which we could successively generate statistically similar versions of the signals at their receivers (with a bounded difference in noise variance) by a linear combination of the signals available at each step, and at the end of the final step, we could decode the messages of all users by initially having the two signals $\{\tilde{Y}_1^n,\tilde{Y}_5^n\}$.

This process can be explained in a more systematic and linear algebraic form. Each of the signals discussed above (ignoring the noise term) can be represented as a $5\times 1$ column vector whose $i^{th}$ element, $i\in[1:5]$, is equal to the coefficient of $g_i^n X_i^n$ in that signal. For instance, $\tilde{Y}_1^n$ corresponds to $\begin{bmatrix}
1 & 0 & 1 & 1 & 0
\end{bmatrix}^T$ and $\tilde{Y}_5^n$ corresponds to $\begin{bmatrix}
0 & 1 & 1 & 0 & 1
\end{bmatrix}^T$. We concatenate these two vectors so that $\{\tilde{Y}_1^n,\tilde{Y}_5^n\}$ can be represented by the matrix
\begin{equation}
\mathbf{A}=\begin{bmatrix}
1 & 0 & 1 & 1 & 0\\
0 & 1 & 1 & 0 & 1
\end{bmatrix}^T.
\end{equation}

Now, using the notation introduced in Definition \ref{d0}, the successive decoding steps mentioned earlier in this example can be expressed in a linear algebraic form. In what follows, $\mathcal{S}=[1:5]$.

\begin{itemize}

\item Step 1 is equivalent to $\mathbf{M}_1\in^\pm \text{span}(\mathbf{A})$. The reason is as follows. First, note that\linebreak$\mathbf{M}_1=\begin{bmatrix}
1&0&1&1&0
\end{bmatrix}^T$ is the first column of the adjacency matrix, which corresponds to the signal received at receiver 1, namely ${Y}_1^n$ (because ${Y}_1^n=\begin{bmatrix} g_{11}^n X_1^n & \dots & g_{51}^n X_5^n \end{bmatrix}\mathbf{M}_1+Z_1^n$). Therefore, $\mathbf{M}_1\in^\pm \text{span}(\mathbf{A})$ means that by a combination of the signals $\{\tilde{Y}_1^n,\tilde{Y}_5^n\}$, we can create a statistically-similar version of ${Y}_1^n$ (actually, the combination is $\tilde{Y}_1^n$ itself) . Since the distribution of the channel gains is symmetric around zero ($f_G(g)=f_G(-g)$, $\forall g\in\mathbb{C}$), the sign of each element $g_{i1}^n X_i^n$ in $Y_1^n$ is not important, therefore letting us use the notation developed in Definition \ref{d0}. In the same way, we have $\mathbf{M}_5\in^\pm \text{span}(\mathbf{A})$, which means that by a combination of the signals $\{\tilde{Y}_1^n,\tilde{Y}_5^n\}$, we can create a statistically-similar version of ${Y}_5^n$.

\item Step 2 is equivalent to $\mathbf{M}_4\in^\pm \text{span}(\mathbf{A},\mathbf{I}_{\{1,5\}}^\mathcal{|S|})$. The reason is as follows. First, note that columns 1 and 5 of the identity matrix are now included since we have already decoded $W_1$ and $W_5$ in the previous step, and by having them and the channel gains, we can create the signals $g_1^n X_1^n$ and $g_5^n X_5^n$ which correspond to $\mathbf{I}_1^\mathcal{|S|}$ and $\mathbf{I}_5^\mathcal{|S|}$, respectively. Therefore, ignoring the noise terms because of their finite variance, we can create a statistically similar version of $Y_4^n$ by having $\tilde{Y}_1^n,\tilde{Y}_5^n,W_1,W_5$ and the channel gains.

\item Step 3 is equivalent to $\mathbf{M}_3\in^\pm \text{span}(\mathbf{A},\mathbf{I}_{\{1,5,4\}}^\mathcal{|S|})$. The reason is as follows. First, note that before this step, we have already decoded $W_1$, $W_5$ and $W_4$, and by having them and the channel gains, we can create the signals $g_1^n X_1^n$, $g_5^n X_5^n$ and $g_4^n X_4^n$ which correspond to $\mathbf{I}_1^\mathcal{|S|}$, $\mathbf{I}_5^\mathcal{|S|}$ and $\mathbf{I}_4^\mathcal{|S|}$, respectively. Therefore, we can create a statistically similar version of $Y_3^n$ by having $\tilde{Y}_1^n,\tilde{Y}_5^n,W_1,W_5,W_4$ and the channel gains.

\item Step 4 is equivalent to $\mathbf{M}_2\in^\pm \text{span}(\mathbf{A},\mathbf{I}_{\{1,5,4,3\}}^\mathcal{|S|})$, which means that we can create a statistically similar version of $Y_2^n$ by having $\tilde{Y}_1^n,\tilde{Y}_5^n,W_1,W_5,W_4,W_3$ and the channel gains.
\end{itemize}
\demo
\end{ex}

Motivated by Example \ref{ex1}, we now formally define the notion of generators.

\begin{defi}\label{d1}
Consider a $K$-user interference network with adjacency matrix $\mathbf{M}$ and assume $\mathcal{S}\subseteq[1:K]$ is a subset of users. $\mathbf{A} \in \{\pm1,0\}^{|\mathcal{S}|\times r}$ ($r\in\mathbb{N}$) is called a \emph{generator} of $\mathcal{S}$ if there exists a sequence $\Pi_\mathcal{S}=(i_1,...,i_{|\mathcal{S}|})$ of the users in $\mathcal{S}$ such that
\begin{align*}
\mathbf{M}_{i_j}^\mathcal{S}&\in_{i_j}^\pm \text{span}\	(\mathbf{A},\mathbf{I}_{\{i_1,...,i_{j-1}\}}^{|\mathcal{S}|}),\	\forall j\in[1:|\mathcal{S}|].
\end{align*}

We use $\mathcal{J}(\mathcal{S})$ to denote the set of all generators of $\mathcal{S}$.
\trig
\end{defi}

To gain intuition about the above definition, as also mentioned in Example \ref{ex1}, each column of a generator $\mathbf{A}$ of $\mathcal{S}$ can be viewed as a representation of a signal which is a linear combination of the transmit symbols $X_i^n$, $i\in\mathcal{S}$. Therefore, the number of columns of $\mathbf{A}$, denoted by $c(\mathbf{A})$, represents the number of these signals. Consequently, the spanning relationships in Definition \ref{d1} represent a sequence of users in which all the messages can be decoded by having $c(\mathbf{A})$ signals, as in Example \ref{ex1}. Also, the reason that we have used the notation $\in_{i_j}^\pm$ instead of $\in^\pm$ (which we were using in Example \ref{ex1}) is that intuitively, it is not necessary to generate (a statistically-similar version of) the received signal at receiver $\text{D}_{i_j}$ \emph{exactly}. Instead, it suffices to generate a \emph{less-interfered} version of its received signal (by deleting some of the interference terms) and still be able to decode its message, because interference only hurts.

By having the definition of the generator in mind, we can present our first converse as follows.
\begin{theo}\label{th1}
The symmetric DoF of a $K$-user interference network with no CSIT is upper bounded by
\begin{equation*}
d_{sym}\leq \underset{\mathcal{S}\subseteq [1:K]}{\emph{min}}\underset{\mathbf{A}\in\mathcal{J}(\mathcal{S})}{\emph{min}} \frac{c(\mathbf{A})}{|\mathcal{S}|},
\end{equation*}
where for each $\mathcal{S}\subseteq[1:K]$, $\mathcal{J}(\mathcal{S})$ denotes the set of all generators of $\mathcal{S}$ (Definition \ref{d1}) and $c(\mathbf{A})$ denotes the number of columns of $\mathbf{A}$.
\end{theo}

Before proving the theorem, we present the following lemma, which is proved in Appendix \ref{apx1}.

\begin{lem}\label{lem2}
For a discrete random variable $W$, continuous random vector $Y^n$, and two complex Gaussian noise vectors $Z_1^n$ and $Z_2^n$, where each element of $Z_1^n$ and $Z_2^n$ are $\mathcal{CN}(0,1)$ and $\mathcal{CN}(0,N)$ random variables, respectively and all the random variables are mutually independent, if $H(W|Y^n+Z_1^n)\leq n\epsilon$, then $H(W|Y^n+Z_2^n)\leq n\epsilon+n\log(N+1)$.
\end{lem}

\begin{proof}[Proof of Theorem \ref{th1}]
Consider a generator of $\mathcal{S}$ denoted by $\mathbf{A}$. Without loss of generality, assume that $\mathcal{S}=[1:m]$, $c(\mathbf{A})=m'$ ($m'\leq m$) and $\Pi_{\mathcal{S}}=(1,...,m)$. Define $\tilde{Y}_i^n=\begin{bmatrix} g_1^n X_1^n & \dots & g_m^n X_m^n \end{bmatrix} \mathbf{A}_i+\tilde{Z}_i^n$, $i\in[1:m']$, where $g_i^n=g_{ii}^n$, $\forall i\in[1:m]$ and the noise vectors $\tilde{Z}_i^n$ have exactly the same distributions as the original noises, but are independent of them and also of each other. Suppose rates $R_i$, $i\in \mathcal{S}$ are achievable. Then, we will have:
\begin{align*}
n\sum_{i\in \mathcal{S}} R_i &= H(W_1,...,W_m|\mathcal{G}^n)\\
&=I(W_1,...,W_m;\tilde{Y}_1^n,...,\tilde{Y}_{m'}^n|\mathcal{G}^n)+H(W_1,...,W_m|\tilde{Y}_1^n,...,\tilde{Y}_{m'}^n,\mathcal{G}^n)\\
&=h(\tilde{Y}_1^n,...,\tilde{Y}_{m'}^n|\mathcal{G}^n)-h(\tilde{Z}_1^n,...,\tilde{Z}_{m'}^n)+H(W_1,...,W_m|\tilde{Y}_1^n,...,\tilde{Y}_{m'}^n,\mathcal{G}^n)\\
&=h(\tilde{Y}_1^n,...,\tilde{Y}_{m'}^n|\mathcal{G}^n)+no(\log(P))+H(W_1,...,W_m|\tilde{Y}_1^n,...,\tilde{Y}_{m'}^n,\mathcal{G}^n).\numberthis\label{eq2}
\end{align*}

Now, we prove that $H(W_l|\tilde{Y}_1^n,...,\tilde{Y}_{m'}^n,W_1,...,W_{l-1},\mathcal{G}^n)\leq no(\log(P))+n\epsilon_{l,n}$ for $l\in[1:m]$. By Definition \ref{d1}, we have that $\mathbf{M}_{l}^\mathcal{S}\in_{l}^\pm \text{span}\	(\mathbf{A},\mathbf{I}_{\{1,...,{l-1}\}}^{|\mathcal{S}|})$. This in turn implies that there exists a vector $\tilde{\mathbf{v}}\in\text{span}\	(\mathbf{A},\mathbf{I}_{\{1,...,{l-1}\}}^{|\mathcal{S}|})$ such that
\begin{itemize}
\item $\tilde{\mathbf{v}}_l$ is either equal to +1 or -1;
\item  $\tilde{\mathbf{v}}_j$ is either equal to 0 or $\pm1$, $\forall j\in\mathcal{IF}_l$; and
\item $\tilde{\mathbf{v}}_j=0$ for all $j\notin \{l\}\cup \mathcal{IF}_l$.
\end{itemize}

This is true because Definition \ref{d0} implies that
\begin{align}
|\tilde{\mathbf{v}}_l|&=|\mathbf{M}_{ll}^\mathcal{S}|=1\label{eqq2}\\
\tilde{\mathbf{v}}_j(|\tilde{\mathbf{v}}_j|-|\mathbf{M}_{jl}^\mathcal{S}|)&=0,\:\forall j\in[1:m]\setminus\{l\}\nonumber,
\end{align}
and we have $\mathbf{M}_{jl}^\mathcal{S}=1$ if $j\in\mathcal{IF}_l$, and $\mathbf{M}_{jl}^\mathcal{S}=0$ if $j\notin\{l\}\cup\mathcal{IF}_l$.

Now, since $\tilde{\mathbf{v}}\in\text{span}\	(\mathbf{A},\mathbf{I}_{\{1,...,{l-1}\}}^{|\mathcal{S}|})$, there exist coefficients $c_i$ ($i\in[1:m']$) and $d_k$ ($k\in[1:l-1]$) such that
\begin{align}
\tilde{\mathbf{v}}&=\sum_{i=1}^{m'} c_i \mathbf{A}_i+\sum_{k=1}^{l-1} d_k \mathbf{I}_k^{|\mathcal{S}|}\label{eqq1}
\end{align}

Multiplying $\begin{bmatrix} g_1^n X_1^n & \dots & g_m^n X_m^n \end{bmatrix}$ by both sides of (\ref{eqq1}), hence, yields
\begin{align*}
\tilde{\mathbf{v}}_l g_l^n X_l^n+\sum_{j\in\mathcal{IF}_l}\tilde{\mathbf{v}}_j g_j^n X_j^n=\sum_{i=1}^{m'} c_i \tilde{Y}_i^n+\sum_{k=1}^{l-1} d_{k} g_k^n X_{k}^n+\tilde{Z'}_l^n,
\end{align*}
where $\tilde{Z'}_l^n=-\sum_{i=1}^{m'} c_i\tilde{Z}_i^n$ and therefore, each of its elements has variance $N_l=\sum_{i=1}^{m'} c_i^2<\infty$. Therefore, we can write:
\begin{align*}
&H\left(W_l|\sum_{i=1}^{m'} c_i \tilde{Y}_i^n+\sum_{k=1}^{l-1} d_{k} g_k^n X_{k}^n,\mathcal{G}^n\right)=
H\left(W_l|\tilde{\mathbf{v}}_l g_l^n X_l^n+\sum_{j\in\mathcal{IF}_l}\tilde{\mathbf{v}}_j g_j^n X_j^n-\tilde{Z'}_l^n,\mathcal{G}^n\right)\\
&\qquad=H\left(W_l|\tilde{\mathbf{v}}_l g_l^n X_l^n+\sum_{j\in\mathcal{IF}_l}\tilde{\mathbf{v}}_j g_j^n X_j^n-\tilde{Z'}_l^n,\sum_{j\in\mathcal{IF}_l}(1-|\tilde{\mathbf{v}}_j|) g_j^n X_j^n,\mathcal{G}^n\right)\numberthis\label{eq6}\\
&\qquad\leq H\left(W_l|\tilde{\mathbf{v}}_l g_l^n X_l^n+\sum_{j\in\mathcal{IF}_l}\tilde{\mathbf{v}}'_j g_j^n X_j^n-\tilde{Z'}_l^n,\mathcal{G}^n\right)\numberthis\label{eq6p}\\
&\qquad\leq no(\log(P))+n\epsilon_{l,n},\numberthis\label{eq7}
\end{align*}
where (\ref{eq6}) is true because, as discussed before, for all $j\in\mathcal{IF}_l$, $\tilde{\mathbf{v}}_j$ can only take the values in $\{\pm1,0\}$ and therefore the signals in $\sum_{j\in\mathcal{IF}_l}\tilde{\mathbf{v}}_j g_j^n X_j^n$ and $\sum_{j\in\mathcal{IF}_l}(1-|\tilde{\mathbf{v}}_j|) g_j^n X_j^n$ do not have common terms.\footnote{If $\tilde{\mathbf{v}}_j=0$, then $1-|\tilde{\mathbf{v}}_j|=1$, and if $\tilde{\mathbf{v}}_j=1$ or $\tilde{\mathbf{v}}_j=-1$, then $1-|\tilde{\mathbf{v}}_j|=0$. Hence, either $\tilde{\mathbf{v}}_j$ or $1-|\tilde{\mathbf{v}}_j|$ is non-zero, but not both.} In (\ref{eq6p}), $\tilde{\mathbf{v}}'_j$ is defined as $\tilde{\mathbf{v}}'_j:=\tilde{\mathbf{v}}_j+(1-|\tilde{\mathbf{v}}_j|)$. Clearly $\tilde{\mathbf{v}}'_j$ can only take the values in $\{+1,-1\}$ because $\tilde{\mathbf{v}}_j\in\{\pm1,0\}$. Also, (\ref{eqq2}) implies that $\tilde{\mathbf{v}}_l\in\{+1,-1\}$. Therefore, $\tilde{\mathbf{v}}_l g_l^n X_l^n+\sum_{j\in\mathcal{IF}_l}\tilde{\mathbf{v}}'_j g_j^n X_j^n-\tilde{Z'}_l^n$ is statistically the same as $Y_l^n$ (with a bounded difference in noise variance), because the channel gains have a symmetric distribution around zero ($f_G(g)=f_G(-g)$, $\forall g\in\mathbb{C}$). This, together with Lemmas \ref{sep} and \ref{lem2} and Fano's inequality, implies that (\ref{eq7}) is correct. Hence, using the chain rule for entropy yields
\begin{align*}
H(W_1,...,W_m|\tilde{Y}_1^n,...,\tilde{Y}_{m'}^n,\mathcal{G}^n)
&=\sum_{l=1}^m H(W_l|\tilde{Y}_1^n,...,\tilde{Y}_{m'}^n,W_1,...,W_{l-1},\mathcal{G}^n)\\
&\leq \sum_{l=1}^m no(\log(P))+n\epsilon_{l,n}\\
&= no(\log(P))+n\epsilon_{n},
\end{align*}
which together with (\ref{eq2}) implies
\begin{align*}
n\sum_{i\in \mathcal{S}} R_i
&\leq h(\tilde{Y}_1^n,...,\tilde{Y}_{m'}^n|\mathcal{G}^n)+no(\log(P))+n\epsilon_{n}\numberthis\label{eqq4}\\
&\leq nm'\log(P)+no(\log(P))+n\epsilon_{n}.
\end{align*}

Letting $n$ and then $P$ go to infinity, we will have:
\begin{equation*}
\sum_{i\in \mathcal{S}} d_i\leq m'\Rightarrow |\mathcal{S}|d_{sym}\leq c(\mathbf{A}) \Rightarrow d_{sym}\leq \frac{c(\mathbf{A})}{|\mathcal{S}|} \Rightarrow d_{sym}\leq \underset{\mathcal{S}\subseteq [1:K]}{\text{min}}\underset{\mathbf{A}\in\mathcal{J}(\mathcal{S})}{\text{min}} \frac{c(\mathbf{A})}{|\mathcal{S}|}.
\end{equation*}
\end{proof}

A simple corollary of Theorem \ref{th1} is the following, which implies that it may be sufficient to only consider as the generators the matrices which are a subset of the columns of the adjacency matrix; i.e. only considering \emph{a subset of the received signals} as our initial signals.

\begin{corl}\label{cor1}
Consider a $K$-user interference network with adjacency matrix $\mathbf{M}$. If $\mathcal{A}\subseteq\mathcal{S}\subseteq[1:K]$ and there exists a sequence $\Pi_{\mathcal{S}\backslash \mathcal{A}}=(i_1,...,i_{|\mathcal{S}\backslash \mathcal{A}|})$ of the users in $\mathcal{S}\backslash \mathcal{A}$ such that:
\begin{align*}
\mathbf{{M}}_{i_j}^\mathcal{S}\in_{i_j}^\pm \emph{span}\	(\mathbf{{M}}_\mathcal{A}^\mathcal{S},\mathbf{I}_{\mathcal{A}\cup\{i_1,...,i_{j-1}\}}^{|\mathcal{S}|}),\	\forall j\in[1:|\mathcal{S}\backslash \mathcal{A}|],
\end{align*}
then $d_{sym}\leq \frac{|\mathcal{A}|}{|\mathcal{S}|}$.
\end{corl}

\begin{proof}
If $\mathcal{A}$ satisfies the conditions in the corollary, then it is easy to show that $\mathbf{{M}}_\mathcal{A}^\mathcal{S}$ is a generator of $\mathcal{S}$ and hence Theorem \ref{th1} yields $d_{sym}\leq \frac{c(\mathbf{{M}}_\mathcal{A}^\mathcal{S})}{|\mathcal{S}|}=\frac{|\mathcal{A}|}{|\mathcal{S}|}$.
\end{proof}

In fact, this corollary can be applied to Example \ref{ex1} to derive the outer bound of $\frac{2}{5}$ for the symmetric degrees-of-freedom. Note that both Theorem \ref{th1} and Corollary \ref{cor1} depend completely on the set of interferers to the receivers or, equivalently, the adjacency matrix. Therefore, they both highlight the special role of the topology of the network on the outer bounds.

It is important to notice that in the final step of the proof of Theorem \ref{th1}, we used the trivial upper bound of $c(\mathbf{A})n\log(P)$ for the joint entropy of the signals corresponding to the generator $\mathbf{A}$. However, there may be a way to derive a tighter upper bound for this joint entropy in some network topologies, and as we see in the next section, this is in fact the case; i.e. there exist some network topologies in which the upper bound of Theorem \ref{th1} can be improved. Hence, in the following, we will illustrate a method to tighten the upper bound.

\subsection{Upper Bounds Based on the Concept of Fractional Generators}

We will now introduce the notion of \emph{fractional generators} to enhance the outer bound of Theorem \ref{th1}. The idea is that we can make use of the signal interactions and interference topology at the receivers to derive possibly tighter upper bounds for the entropy of the signals corresponding to a generator. To be precise, if a signal is composed of a subset of interferers to a receiver, there is a tighter upper bound than $n\log(P)$ for that signal. To clarify this concept, we will again go through an introductory example.

\begin{ex}\label{ex2}
Consider the 6-user network shown in Figure \ref{fig2}. We claim that the symmetric DoF for this network is upper bounded by $\frac{2}{7}$, while the best upper bound based on Theorem \ref{th1} is $\frac{2}{6}$.
\begin{figure}[h]
\centering
\includegraphics[trim = 2in 3in 2in 3.5in, clip,width=0.3\textwidth]{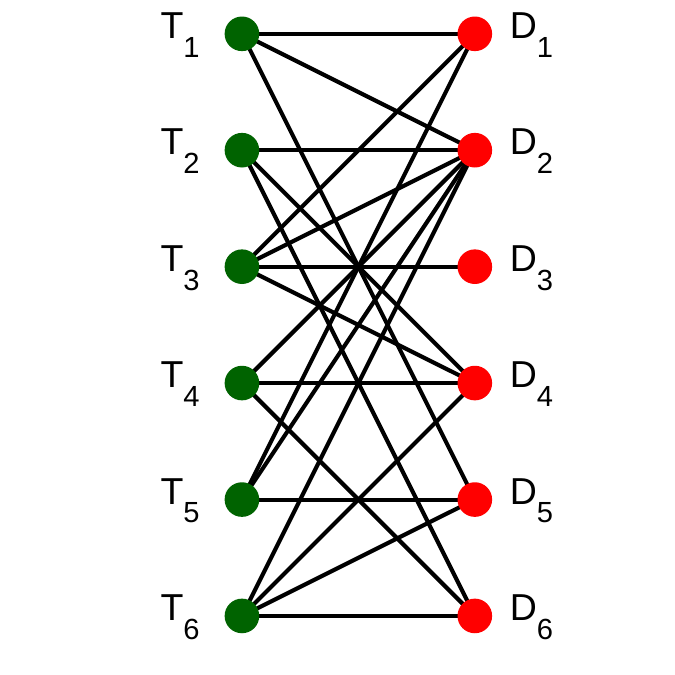}
\caption{A 6-user interference network in which the upper bound of Theorem \ref{th1} is not tight.}
\label{fig2}
\end{figure}

The best upper bound of Theorem \ref{th1} for this example can be shown to be $\frac{2}{6}$, which is obtained by, for example, using $\mathbf{A}=\mathbf{M}_{\{1,4\}}^\mathcal{S}$ as a generator of the entire network $\mathcal{S}=[1:6]$ with $\Pi_{\mathcal{S}}=\{1,4,2,5,3,6\}$. We now show how the proof steps of Theorem \ref{th1} can be enhanced to obtain a tighter upper bound on $d_{sym}$.

Following the proof of Theorem \ref{th1} for the network in Figure \ref{fig2} until equation (\ref{eqq4}) provides
\begin{equation}\label{eqq5}
n\sum_{i=1}^6 R_i \leq h(\tilde{Y}_1^n,\tilde{Y}_4^n|\mathcal{G}^n)+no(\log(P))+n\epsilon_n,
\end{equation}
where $\tilde{Y}_1^n=g_1^nX_1^n+g_3^nX_3^n+g_5^nX_5^n+\tilde{Z}_1^n$ and $\tilde{Y}_4^n=g_2^nX_2^n+g_3^nX_3^n+g_4^nX_4^n+g_6^nX_6^n+\tilde{Z}_4^n$, $\tilde{Z}_1^n$ and $\tilde{Z}_4^n$ have the same distributions as the original noise vectors, but are independent of them and also of each other and $g_i^n=g_{ii}^n$, $i\in[1:6]$. Now, instead of simply upper bounding $h(\tilde{Y}_1^n,\tilde{Y}_4^n|\mathcal{G}^n)$ as $h(\tilde{Y}_1^n,\tilde{Y}_4^n|\mathcal{G}^n)\leq h(\tilde{Y}_1^n|\mathcal{G}^n)+h(\tilde{Y}_4^n|\mathcal{G}^n)\leq 2n\log(P)+no(\log(P))$, we show that a tighter upper bound can be found for $h(\tilde{Y}_1^n|\mathcal{G}^n)$, hence improving the upper bound on $d_{sym}$.

The idea is that in the network of Figure \ref{fig2}, $\text{D}_1$ receives signals from transmitters 1, 3 and 5. However, these transmitters are a \emph{subset of the interferers to receiver 2}; i.e. $\{1,3,5\}\subseteq\mathcal{IF}_2$. This leads to a tighter upper bound of $h(\tilde{Y}_1^n|\mathcal{G}^n)\leq n(\log(P)-R_2)+no(\log(P))+n\epsilon_n$, which can be proved as follows.

First, note that the following equality is true.
\begin{align*}
H(W_2)-H(W_2|g_{22}^n X_2^n+\tilde{Y}_1^n,\mathcal{G}^n)=h(g_{22}^n X_2^n+\tilde{Y}_1^n|\mathcal{G}^n)-h(g_{22}^n X_2^n+\tilde{Y}_1^n|W_2,\mathcal{G}^n),
\end{align*}
because the two sides are equivalent expressions of $I(g_{22}^n X_2^n+\tilde{Y}_1^n;W_2|\mathcal{G}^n)$. Therefore, we have
\begin{align*}
h(g_{22}^n X_2^n+\tilde{Y}_1^n|W_2,\mathcal{G}^n)
&=H(W_2|g_{22}^n X_2^n+\tilde{Y}_1^n,\mathcal{G}^n)+h(g_{22}^n X_2^n+\tilde{Y}_1^n|\mathcal{G}^n)-H(W_2)\\
&\leq n\epsilon_n+h(g_{22}^n X_2^n+\tilde{Y}_1^n|\mathcal{G}^n)-H(W_2)\numberthis\label{eq9}\\
&\leq n\epsilon_n+no(\log(P))+n\log(P)-n R_2,\numberthis\label{eq10}
\end{align*}
where (\ref{eq9}) holds because of Fano's inequality and the fact that $g_{22}^n X_2^n+\tilde{Y}_1^n$ is a less-interfered version of the received signal at receiver 2 and is able to decode $W_2$ due to Lemma \ref{sep}.

On the other hand, since $X_2^n$ is a function of $W_2$, we have
\begin{equation*}
h(g_{22}^n X_2^n+\tilde{Y}_1^n|W_2,\mathcal{G}^n)=h(\tilde{Y}_1^n|\mathcal{G}^n),
\end{equation*}
which together with (\ref{eq10}) yields $h(\tilde{Y}_1^n|\mathcal{G}^n)\leq n(\log(P)-R_2)+no(\log(P))+n\epsilon_n$. Hence we can continue (\ref{eqq5}) as
\begin{align*}
n\sum_{i=1}^6 R_i &\leq 
h(\tilde{Y}_1^n|\mathcal{G}^n)+h(\tilde{Y}_4^n|\mathcal{G}^n)+no(\log(P))+n\epsilon_n\\
&\leq 2n\log(P)-nR_2+no(\log(P))+n\epsilon_n.
\end{align*}

Letting $n$ and then $P$ go to infinity and setting all the DoFs to be equal to $d_{sym}$, we will have:
\begin{equation*}
6d_{sym}\leq 2-d_{sym}\Rightarrow d_{sym}\leq \frac{2}{7},
\end{equation*}
which is strictly tighter than the previous outer bound of $\frac{2}{6}$ based on Theorem \ref{th1}.

Now, we will illustrate the improvement of the outer bound in a linear algebraic form. The key part in the enhancement was that by adding $g_{22}^nX_2^n$ to $\tilde{Y}_1^n$, we could create a signal which was able to decode $W_2$. As we have discussed before, if $\mathcal{S}=[1:6]$, then $\tilde{Y}_1^n$ corresponds to the vector $\mathbf{M}_1^\mathcal{S}=\begin{bmatrix}
1 & 0 & 1 & 0 & 1 & 0
\end{bmatrix}^T$. Therefore, adding $g_{22}^nX_2^n$ to $\tilde{Y}_1^n$ can be translated to adding $\mathbf{I}_2^{|\mathcal{S}|}=\begin{bmatrix}
0 & 1 & 0 & 0 & 0 & 0
\end{bmatrix}^T$ to $\mathbf{M}_1^\mathcal{S}$. Moreover, the fact that $W_2$ can be decoded from $g_{22}^nX_2^n+\tilde{Y}_1^n$ is equivalent to $\mathbf{M}_2^\mathcal{S}\in_2^\pm\text{span}(\mathbf{M}_1^\mathcal{S}+\mathbf{I}_2^{|\mathcal{S}|})$. We will call $\mathbf{M}_1^\mathcal{S}$ a \emph{fractional generator} of $\mathcal{S}'$ in $\mathcal{S}$ where $\mathcal{S}'=\{2\}$. This means that by expanding the signal corresponding to $\mathbf{M}_1^\mathcal{S}$ (through adding $\mathbf{I}_2^{|\mathcal{S}|}$ to $\mathbf{M}_1^\mathcal{S}$ or equivalently $g_{22}^nX_2^n$ to $\tilde{Y}_1^n$), the resulting expanded signal is able to decode $W_2$. This is the method that we will use to linear algebraically describe the improvement in the outer bound on $d_{sym}$.
\demo
\end{ex}

\begin{remk}
A similar approach in \cite{jafar} has been taken to derive an upper bound for the symmetric DoF of general network topologies. In particular, if for the network in Example \ref{ex2}, we set $h_{26}=h_{46}=h_{66}=-\sqrt{\text{SNR}\times \frac{N_0}{P}}$ and the other channel gains $h_{ji}$ to $\sqrt{\text{SNR}\times \frac{N_0}{P}}$, then \emph{maximum cardinality of an acyclic subset of messages}, denoted by $\Psi$, is equal to 3 and the \emph{minimum internal conflict distance}, denoted by $\Delta$, is equal to 1 for the network of Figure \ref{fig2}. Therefore, both of the bounds presented in Theorem 4.12 and Corollary 4.13 of \cite{jafar} for the network of Figure \ref{fig2} are equal to $\frac{1}{3}$, while the outer bound of $\frac{2}{7}$ that we derived in Example \ref{ex2} is strictly tighter.
\end{remk}

To generalize the improvement of the outer bound to all network topologies, we define the concept of \emph{fractional generator}.

\begin{defi}\label{d2}
Consider a $K$-user interference network with adjacency matrix $\mathbf{M}$ and suppose $\mathcal{S'}\subseteq\mathcal{S}\subseteq[1:K]$. A vector $\mathbf{c} \in \{\pm1,0\}^{|\mathcal{S}|}$ is called a \emph{fractional generator of $\mathcal{S'}$ in $\mathcal{S}$} if $\mathbf{c}_k=0, \forall k\in\mathcal{S}'$ and there exists a sequence $\Pi_\mathcal{S'}=(i_1,...,i_{|\mathcal{S'}|})$ of the users in $\mathcal{S'}$ such that:
\begin{align*}
\mathbf{M}_{i_j}^\mathcal{S}&\in_{i_j}^\pm \text{span}\	\left(\mathbf{c}+\sum_{k\in\mathcal{S'}} \mathbf{I}_k^{|\mathcal{S}|},\mathbf{I}_{\{i_1,...,i_{j-1}\}}^{|\mathcal{S}|}\right),\	\forall j\in[1:|\mathcal{S'}|].
\end{align*}

We use the notation $\mathcal{J}_{\mathcal{S}}(\mathcal{S'})$ to denote the set of all fractional generators of $\mathcal{S'}$ in $\mathcal{S}$.
\trig
\end{defi}

Intuitively, a fractional generator of $\mathcal{S'}$ in $\mathcal{S}$ is a column vector whose corresponding signal can decode the messages of the users in $\mathcal{S}'$ (which is a subset of the set of entire users $\mathcal{S}$) sequentially, after expansion by adding $\sum_{k\in\mathcal{S'}} \mathbf{I}_k^{|\mathcal{S}|}$ to it (or equivalently, by adding $\sum_{k\in\mathcal{S'}} g_{kk}^n X_k^n$ to its corresponding signal).

After having the definition of fractional generators, we can state the following lemma, which is proved in Appendix \ref{apx2}.

\begin{lem}\label{lem3}
Consider a subset of users $\mathcal{S}\subseteq[1:K]$ in a $K$-user interference network and suppose $\mathbf{c}\in\mathcal{J}_{\mathcal{S}}(\mathcal{S'})$, where $\mathcal{S'}$ is a subset of $\mathcal{S}$. If rates $R_i$ are achievable for all $i\in\mathcal{S'}$, then
\begin{equation*}
h\left(\sum_{j\in\mathcal{S}} \mathbf{c}_j g_j^n X_j^n+Z^n|\mathcal{G}^n\right)\leq n\left(\log(P)-\sum_{i\in\mathcal{S'}} R_i\right)+no(\log(P))+n\epsilon_n,
\end{equation*}
where $g_j^n=g_{jj}^n$ , $\forall j\in\mathcal{S}$, $\epsilon_n\rightarrow 0$ as $n\rightarrow\infty$, and $Z^n$ is a white Gaussian noise vector with each element distributed as $\mathcal{CN}(0,1)$, independent of the transmit symbols and the channel gains.
\end{lem}

Moreover, for a vector $\mathbf{c}\in \{\pm1,0\}^{|\mathcal{S}|}$, we define $n_{\mathcal{S}}(\mathbf{c})$ as the size of the largest subset $\mathcal{S'}$ of $\mathcal{S}$ such that $\mathbf{c}$ is a fractional generator of $\mathcal{S}'$ in $\mathcal{S}$. To be precise, we have the following definition.

\begin{defi}\label{d3}
Consider a subset of users $\mathcal{S}\subseteq[1:K]$ in a $K$-user interference network. For a vector $\mathbf{c}\in \{\pm1,0\}^{|\mathcal{S}|}$, $n_{\mathcal{S}}(\mathbf{c})$ is defined as
\begin{align*}
n_{\mathcal{S}}(\mathbf{c}):= \underset{\mathcal{S'}\subseteq\mathcal{S}}{\text{max}}\:&|\mathcal{S'}|\\
\text{s.t. }&\mathbf{c}\in\mathcal{J}_{\mathcal{S}}(\mathcal{S'}).
\end{align*}
\trig
\end{defi}

Note that, due to Lemma \ref{lem3}, finding $n_{\mathcal{S}}(\mathbf{c})$ leads to the tightest upper bound for the signal corresponding to $\mathbf{c}$. Therefore, we are now at a stage to state our second converse.

\begin{theo}\label{th2}
The symmetric DoF of a $K$-user interference network with no CSIT is upper bounded by
\begin{equation*}
d_{sym}\leq \underset{\mathcal{S}\subseteq [1:K]}{\emph{min}}\underset{\mathbf{A}\in\mathcal{J}(\mathcal{S})}{\emph{min}} \frac{c(\mathbf{A})}{|\mathcal{S}|+\sum_{i=1}^{c(\mathbf{A})}n_{\mathcal{S}}(\mathbf{A}_i)},
\end{equation*}
where for each $\mathcal{S}\subseteq[1:K]$, $\mathcal{J}(\mathcal{S})$ denotes the set of all generators of $\mathcal{S}$ (Definition \ref{d1}), $c(\mathbf{A})$ denotes the number of columns of $\mathbf{A}$  and $n_{\mathcal{S}}(\mathbf{A}_i)$ is defined as in Definition \ref{d3}.
\end{theo}

\begin{proof}
Following the proof of Theorem \ref{th1} until equation (\ref{eqq4}), we know that if $\mathcal{S}=\{1,...,m\}$, $c(\mathbf{A})=m'$, $\Pi_{\mathcal{S}}=(1,...,m)$, $\tilde{Y}_i^n=\begin{bmatrix} g_1^n X_1^n & \dots & g_m^n X_m^n \end{bmatrix} \mathbf{A}_i+\tilde{Z}_i^n$, $i\in[1:m']$ and if rates $R_i$ ($i\in \mathcal{S}$) are achievable, we will have
\begin{align*}
n\sum_{i\in \mathcal{S}} R_i &\leq
h(\tilde{Y}_1^n,...,\tilde{Y}_{m'}^n|\mathcal{G}^n)+no(\log(P))+n\epsilon_{n}\\
&\leq \sum_{i=1}^{m'} h(\tilde{Y}_i^n|\mathcal{G}^n)+no(\log(P))+n\epsilon_{n}.\numberthis\label{eq16}
\end{align*}

Now, if $\mathbf{A}_i\in\mathcal{J}_{\mathcal{S}}(\mathcal{S'})$, then Lemma \ref{lem3} implies
\begin{align*}
n\left(\log(P)-\sum_{j\in\mathcal{S'}} R_j\right)+no(\log(P))+n\epsilon_n
&\geq h\left(\sum_{j\in\mathcal{S}} \mathbf{A}_{ji} g_j^n X_j^n+\tilde{Z}_i^n|\mathcal{G}^n\right) \\
&= h\left(\begin{bmatrix} g_1^n X_1^n & \dots & g_m^n X_m^n \end{bmatrix} \mathbf{A}_i+\tilde{Z}_i^n|\mathcal{G}^n\right)\\
&= h(\tilde{Y}_i^n|\mathcal{G}^n).\numberthis\label{eq17}
\end{align*}

Thus, to find the tightest upper bound on $h(\tilde{Y}_i^n|\mathcal{G}^n)$ for every $i\in[1:c(\mathbf{A})]$, we need to find the largest subset $\mathcal{S'}$ such that $\mathbf{A}_i\in\mathcal{J}_{\mathcal{S}}(\mathcal{S'})$, which we denote by $\mathcal{S'}_i^*$; i.e. $\mathcal{S'}_i^*=\text{arg }\underset{\mathcal{S'}}{\text{max}}\:|\mathcal{S'}|$ s.t. $\mathbf{A}_i\in\mathcal{J}_{\mathcal{S}}(\mathcal{S'})$. Combining this with (\ref{eq16}) and (\ref{eq17}) yields
\begin{align*}
n\sum_{i\in \mathcal{S}} R_i
&\leq \sum_{i=1}^{c(\mathbf{A})} n(\log(P)-\sum_{j\in\mathcal{S'}_i^*} R_j)+no(\log(P))+n\epsilon_n.
\end{align*}

Letting $n$ and then $P$ go to infinity and setting all the DoFs to be equal to $d_{sym}$, we will have:
\begin{align*}
|\mathcal{S}|d_{sym}&\leq c(\mathbf{A})- \sum_{i=1}^{c(\mathbf{A})} |\mathcal{S'}_i^*|d_{sym} = c(\mathbf{A})- \sum_{i=1}^{c(\mathbf{A})} n_{\mathcal{S}}(\mathbf{A}_i)d_{sym}\\
\Rightarrow d_{sym}&\leq \frac{c(\mathbf{A})}{|\mathcal{S}|+\sum_{i=1}^{c(\mathbf{A})}n_{\mathcal{S}}(\mathbf{A}_i)}\\
\Rightarrow
d_{sym}&\leq \underset{\mathcal{S}\subseteq [1:K]}{\text{min}}\underset{\mathbf{A}\in\mathcal{J}(\mathcal{S})}{\text{min}} \frac{c(\mathbf{A})}{|\mathcal{S}|+\sum_{i=1}^{c(\mathbf{A})}n_{\mathcal{S}}(\mathbf{A}_i)}.
\end{align*}
\end{proof}

As it is clear from the above discussion, the outer bound of Theorem \ref{th2} captures the impact of network topology on upper bounding the symmetric DoF more strongly than Theorem \ref{th1}. In fact, Theorem \ref{th2} tries to focus on the signal and interference interactions at the receivers through Lemma \ref{lem3}, which is the key aspect of the improvement of the bound compared to the bound suggested by Theorem \ref{th1}.

\section{Inner Bounds on $d_{sym}$}\label{ach}

In this section, we derive inner bounds on the symmetric degrees-of-freedom. In particular, we focus on two benchmark schemes, namely \emph{random Gaussian coding} and \emph{interference avoidance}, and introduce a new scheme called \emph{structured repetition coding}. The structured repetition coding scheme in general performs better than (or at least the same as) the first two schemes and as we illustrate in Section \ref{sim}, it closes the gap between the inner and outer bounds in many networks where the first two schemes fail to do so.

\subsection{Benchmark Schemes}\label{bench}

We start by presenting two benchmark schemes and we will compare them with each other through examples to study their performance with respect to our outer bounds in Section \ref{converse}.

\subsubsection{Random Gaussian Coding and Interference Decoding}

In the first scheme, we use random Gaussian coding, such that all interfering messages at each receiver are decoded. Consider a $K$-user interference network and look at one of the receivers, say $\text{D}_j$. It receives signals from $\text{T}_i$, $i\in\{j\}\cup\mathcal{IF}_j$. Therefore, we can see this subnetwork as a multiple access channel (MAC) to receiver $j$. It is well known \cite{tse} that in the fast fading settings, the capacity region of MAC with no CSIT is specified by
\begin{align*}
\sum_{i\in\mathcal{S}} R_i\leq \mathbb{E}\left[\log\left( 1+\sum_{i\in\mathcal{S}}|g_{ij}|^2P\right) \right],\:\forall\mathcal{S}\subseteq\{j\}\cup\mathcal{IF}_j,
\end{align*}
where the expectation is taken with respect to the channel gains. Now, since $\log(x)<\log(1+x)$ for all positive $x$, the rates $R_i$ are achievable if they satisfy
\begin{align*}
\sum_{i\in\mathcal{S}} R_i\leq \mathbb{E}\left[\log\left(\sum_{i\in\mathcal{S}}|g_{ij}|^2P\right) \right]=\log(P)+\mathbb{E}\left[\log\left(\sum_{i\in\mathcal{S}}|g_{ij}|^2\right) \right],\:\forall\mathcal{S}\subseteq\{j\}\cup\mathcal{IF}_j.
\end{align*}

Then, because $\log(.)$ is a monotonically increasing function, the rates $R_i$ are achievable if the following holds.
\begin{align*}
\sum_{i\in\mathcal{S}} R_i\leq \log(P)+\mathbb{E}\left[\log\left(|g_{i_\mathcal{S}j}|^2\right) \right],\:\forall\mathcal{S}\subseteq\{j\}\cup\mathcal{IF}_j,
\end{align*}
where for each $\mathcal{S}\subseteq\{j\}\cup\mathcal{IF}_j$, $i_\mathcal{S}$ is some user in $\mathcal{S}$. From the regularity conditions on the distribution of the channel gains (mentioned in Section \ref{form_not}), it can be shown that $\mathbb{E}\left[\log\left( |g|^2\right) \right]>-\infty$ (a more general case is proved in Appendix \ref{apx4}). Therefore, dividing the above equations by $\log(P)$ and letting $P$ go to infinity leads to the fact that the degrees-of-freedom $d_j$ are achievable if
\begin{align*}
\sum_{i\in\mathcal{S}} d_i\leq 1,\:\forall\mathcal{S}\subseteq\{j\}\cup\mathcal{IF}_j.
\end{align*}

For the degrees-of-freedom to be symmetric, we will therefore have $d_{sym}\leq\frac{1}{|\mathcal{S}|}$ which should hold for every $\mathcal{S}\subseteq\{j\}\cup\mathcal{IF}_j$. Choosing the largest subset $\mathcal{S}$ yields $d_{sym}\leq\frac{1}{1+|\mathcal{IF}_j|}$. Furthermore, all the rates (degrees-of-freedom) in this region can be achieved using random Gaussian codebooks of size $2^{nR_i}\times n$ generated for each user, in which all the elements are i.i.d. $\mathcal{CN}(0,P)$. The message $W_i$ is the index of the row of this codebook matrix and the transmit vector will be the corresponding row of the codebook. Therefore, by applying the viewpoint of multiple access channels to all the receivers in the interference network, this theorem follows immediately.

\begin{theo}\label{th3}
Consider a $K$-user interference network with adjacency matrix $\mathbf{M}$. If we denote the maximum receiver degree by $\Delta_R$ (defined as $\Delta_R:=1+\underset{j\in[1:K]}{\max}|\mathcal{IF}_j|=\underset{j\in[1:K]}{\max}\sum_{i=1}^K \mathbf{M}_{ij}$), then the symmetric DoF of $\frac{1}{\Delta_R}$ is achievable.
\end{theo}

Theorem \ref{th3} only considers the maximum degree among the receivers to derive an inner bound on $d_{sym}$. However, it fails to capture how further details of network topology can affect the achievable symmetric DoF. In other words, this theorem suggests a similar inner bound for all network topologies whose maximum receiver degrees are identical, implying its possible suboptimality for many networks. Therefore, we should seek for other schemes that exploit other structures in the network topology.

\subsubsection{Interference Avoidance}

As the name suggests, this scheme is based on avoiding the interference by all the users. Each transmitter, aware of the network topology, knows the receivers which receive interference from itself and also the transmitters who cause interference at its corresponding receiver. Therefore, it can avoid sending its symbols at the same time as those users. In other words, in this scheme, each user uses a time slot to transmit data if and only if the users who receive interference from/cause interference at that user do not use that time slot. This is tightly connected to the concept of \emph{independent sets}.

Suppose we have a $K$-user interference network. $\mathcal{U}\subseteq[1:K]$ is an independent set if for all two distinct users $i$ and $j$ in $\mathcal{U}$, $\mathbf{M}_{ij}=\mathbf{M}_{ji}=0$; i.e. users $i$ and $j$ are mutually non-interfering. Obviously, all the users in an independent set can transmit their symbols at the same time without experiencing any interference. This is the essence of the \emph{interference avoidance} scheme. Naturally, it is best if the largest possible subset of the users send together, leading to the concept of \emph{maximal} independent sets. $\mathcal{U}$ is a maximal independent set if it is an independent set, but for all $l\in[1:K]\backslash\mathcal{U}$, $\mathcal{U}\cup\{l\}$ is not an independent set.

After describing the above scheme, we can state our second inner bound.

\begin{theo}\label{th4}
Consider a $K$-user interference network with adjacency matrix $\mathbf{M}$ and suppose $\mathcal{U}=\{\mathcal{U}_1,...,\mathcal{U}_m\}$ is the set of all maximal independent sets of this network. Then, the following symmetric DoF is achievable by interference avoidance.
\begin{align*}
\underset{n\in\mathbb{N}}{\emph{sup}}\:\underset{\mathcal{U'}_1,...,\mathcal{U'}_n\in\mathcal{U}}{\emph{max}}\: \underset{i\in[1:K]}{\emph{min}}\frac{\sum_{j=1}^n \mathbf{1}(i\in\mathcal{U'}_j)}{n},
\end{align*}
where for an event $A$, $\mathbf{1}(\text{A})=1$ if A occurs and $\mathbf{1}(\text{A})=0$ otherwise.
\end{theo}

\begin{proof}
If we take $n$ maximal independent sets $\mathcal{U'}_1,...,\mathcal{U'}_n$ and allow all the users in $\mathcal{U'}_j$ to transmit simultaneously in time slot $j$, $j\in[1:n]$, then for every user $i$, $i\in[1:K]$, there will be $\sum_{j=1}^n \mathbf{1}(i\in\mathcal{U'}_j)$ clean, interference-free, channels between $\text{T}_i$ and $\text{D}_i$. Hence, each user $i$ achieves $\frac{\sum_{j=1}^n \mathbf{1}(i\in\mathcal{U'}_j)}{n}$ degrees-of-freedom. Since we are interested in the achievable symmetric degrees-of-freedom, the maximum DoF that all the users can simultaneously achieve with a specific choice of $n$ and $\mathcal{U'}_1,...,\mathcal{U'}_n$ is $\underset{i\in[1:K]}{\text{min}}\frac{\sum_{j=1}^n \mathbf{1}(i\in\mathcal{U'}_j)}{n}$. Optimizing over $n$ and $\mathcal{U'}_1,...,\mathcal{U'}_n$, the best symmetric DoF achievable under interference avoidance is $\underset{n\in\mathbb{N}}{\text{sup}}\:\underset{\mathcal{U'}_1,...,\mathcal{U'}_n\in\mathcal{U}}{\text{max}}\: \underset{i\in[1:K]}{\text{min}}\frac{\sum_{j=1}^n \mathbf{1}(i\in\mathcal{U'}_j)}{n}$.
\end{proof}

\begin{remk}
The aforementioned ideas of independent sets are very closely related to fractional coloring and fractional chromatic numbers of graphs in graph theory \cite{fgt}. To relate the two problems, we define the \emph{conflict graph} of a $K$-user interference network with adjacency matrix $\mathbf{M}$ as an undirected graph $G=(\mathcal{V},\mathcal{E})$  with the set of vertices $\mathcal{V}=[1:K]$ and the set of edges $\mathcal{E}$ where for all $i\neq j$, $e_{ij}\in \mathcal{E} $ if $\mathbf{M}_{ij}=1$ or $\mathbf{M}_{ji}=1$ in the original interference network. Now, the assignment of time slots to different users based on independent sets corresponds to \emph{coloring} the conflict graph $G$. An $n$-\emph{coloring} of a graph $G=(\mathcal{V},\mathcal{E})$ is an assignment of a single color out of a set of $n$ colors to each of the vertices in $\mathcal{V}$ such that if $e_{ij}\in \mathcal{E}$, different colors are assigned to vertices $i$ and $j$. The smallest $n$ for which an $n$-coloring is possible for $G$ is called the \emph{chromatic number} of $G$, denoted by $\chi(G)$.

Moreover, an $m$-\emph{fold coloring} (known as \emph{fractional coloring}) of a graph $G$ is an assignment of sets of $m$ colors to each vertex in $\mathcal{V}$ such that if $e_{ij}\in \mathcal{E}$, the sets of colors assigned to vertices $i$ and $j$ are disjoint. Also, $G$ is said to be $n:m$-\emph{colorable} if there exists an $m$-fold coloring of $G$ such that all the colors used in the coloring are drawn from a set of $n$ distinct colors. The smallest $n$ for which $G$ is $n:m$-colorable is called the $m$-fold chromatic number of $G$, denoted by $\chi_m(G)$. The maximum symmetric DoF achievable by interference avoidance is $\underset{m\in\mathbb{N}}{\text{sup}}\frac{m}{\chi_m(G)}$ which is exactly the value presented in Theorem \ref{th4}.
\footnote{For every $m\in\mathbb{N}$, interference avoidance can achieve the symmetric DoF of $\frac{m}{\chi_m(G)}$, because for every $m$-fold chromatic number $\chi_m(G)$, $m$ is the largest $\bar{m}$ such that an $\bar{m}$-fold coloring exists for $G$, where the colors are selected out of a palette of $\chi_m(G)$ colors. Each color out of the total of $\chi_m(G)$ colors corresponds to an independent set. Hence, $m$ is the maximum $\bar{m}$ such that each node appears $\bar{m}$ times in the independent sets corresponding to $\chi_m(G)$ colors. In other words, if $\mathcal{U}$ is the set of all maximal independent sets of the interference network, then
\begin{equation*}
m=\underset{\mathcal{U'}_1,...,\mathcal{U'}_{\chi_m(G)}\in\mathcal{U}}{\max}\: \underset{i\in[1:K]}{\min}\sum_{j=1}^{\chi_m(G)} \mathbf{1}(i\in\mathcal{U'}_j),
\end{equation*}
because each user appears at least $\underset{i\in[1:K]}{\min}\sum_{j=1}^{\chi_m(G)} \mathbf{1}(i\in\mathcal{U'}_j)$ times among the independent sets $\mathcal{U'}_1,...,\mathcal{U'}_{\chi_m(G)}$ and $m$ is the maximum value of this quantity where the maximization is over the selection of independent sets corresponding to $\chi_m(G)$ colors. Optimizing over $m$ yields the inner bound in Theorem \ref{th4}.}
However, the \emph{fractional chromatic number} of $G$ is defined as $\chi_f(G)=\underset{m\in\mathbb{N}}{\text{inf}}\frac{\chi_m(G)}{m}$, which can also be shown to equal $\underset{m\rightarrow\infty}{\text{lim}}\frac{\chi_m(G)}{m}$ \cite{fgt}. Therefore, the best symmetric DoF achievable by interference avoidance is in fact $\frac{1}{\chi_f(G)}$.
\end{remk}

The two schemes we presented so far, incorporate two different aspects of network topology, namely maximum receiver degree and fractional chromatic number, to improve spectral efficiency. A natural question that comes to mind is: How do these two schemes compare to each other? Is one of them superior than the other one for all network graphs? The answer is negative. We will present two examples to clarify how the schemes work and also to compare them. In the first example, random Gaussian coding performs better, while in the second one, interference avoidance outperforms the first scheme.

\begin{ex}

Consider the 4-user network in Figure \ref{fig3a}.
\begin{figure}
\centering
\begin{subfigure}[b]{0.45\textwidth}
\centering
\includegraphics[trim = 2in 3.2in 2in 5.1in, clip,width=0.6\textwidth]{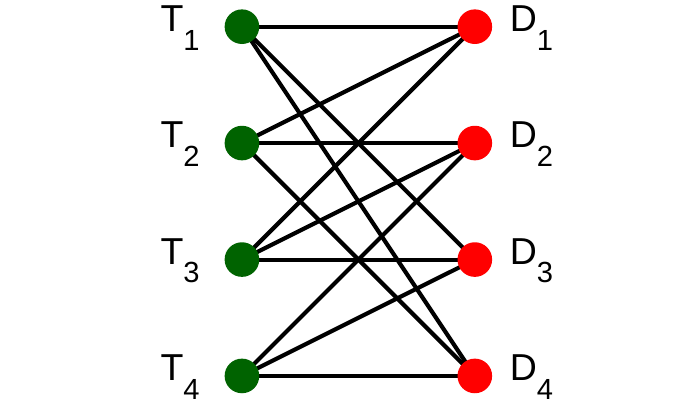}
\caption{}
\label{fig3a}
\end{subfigure}
~
\begin{subfigure}[b]{0.45\textwidth}
\centering
\includegraphics[trim = 2in 3in 2in 2.73in, clip,width=0.8\textwidth]{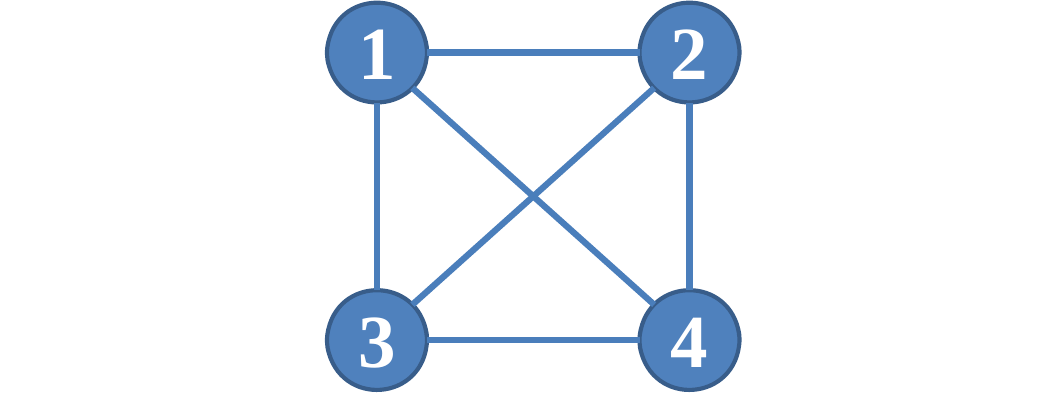}
\caption{}
\label{fig3b}
\end{subfigure}
\caption{(a) A 4-user interference network in which random Gaussian coding is optimal, and (b) its corresponding conflict graph.}
\end{figure}
Suppose we want to apply interference avoidance to this network. We should identify the independent sets in this network. Clearly, all the users are mutually interfering in this network. This can also be seen in the fully connected conflict graph of Figure \ref{fig3b}, whose maximal independent sets are $\{1\},\{2\},\{3\}$ and $\{4\}$, implying that the best symmetric DoF achievable under interference avoidance is $\frac{1}{4}$. However, the maximum receiver degree in this network is $\Delta_R=3$ and therefore, Theorem \ref{th3} implies that random Gaussian coding and interference decoding can achieve the symmetric DoF of $\frac{1}{3}$ which is higher than the value achieved by interference avoidance.

To show that the symmetric DoF of $\frac{1}{3}$ is optimal, it is necessary to mention the outer bound, too. If you consider the subnetwork consisting of the users $\mathcal{S}=\{1,2,3\}$, then clearly $[1\:1\:1]^T$ is a generator of $\mathcal{S}$. Therefore, using Theorem \ref{th1}, $d_{sym}\leq\frac{1}{3}$ implying the optimality of random Gaussian coding and interference decoding in this network, whereas interference avoidance performs suboptimally in this case.
\demo
\end{ex}

\begin{ex}
As our next example, we return to the network we considered in Example \ref{ex1}, which is repeated in Figure \ref{fig4a} for convenience.
\begin{figure}[hbt]
\centering
\begin{subfigure}[b]{0.31\textwidth}
\centering
\includegraphics[trim = 2in 3in 2in 4.1in, clip,width=0.8\textwidth]{fig1}
\caption{}
\label{fig4a}
\end{subfigure}
~
\begin{subfigure}[b]{0.31\textwidth}
\centering
\includegraphics[trim = 2in 1.2in 2in 1.4in, clip,width=0.85\textwidth]{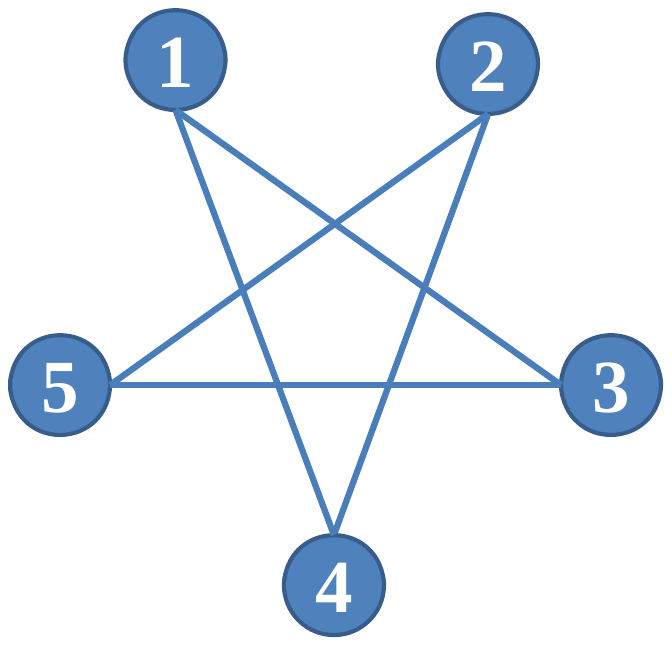}
\caption{}
\label{fig4b}
\end{subfigure}
~
\begin{subfigure}[b]{0.31\textwidth}
\centering
\includegraphics[trim = 2in 1.2in 2in 1.4in, clip,width=0.85\textwidth]{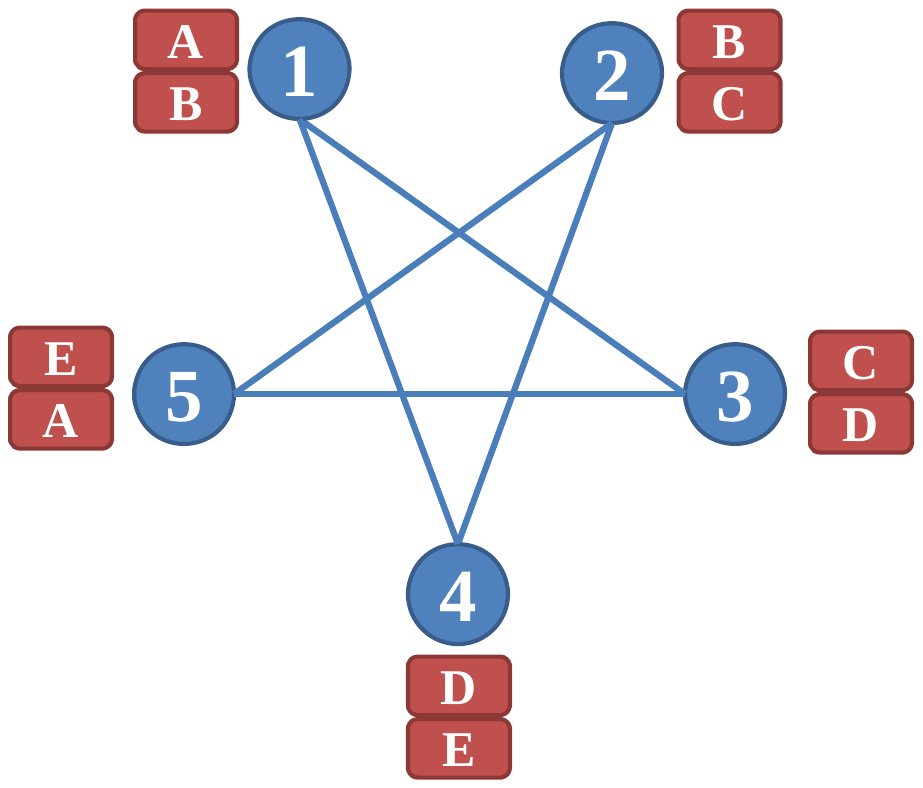}
\caption{}
\label{fig4c}
\end{subfigure}
\caption{(a) A 5-user interference network in which interference avoidance is optimal, (b) the corresponding conflict graph, and (c) a 5:2-coloring.}
\end{figure}

As shown before, for this network $d_{sym}\leq\frac{2}{5}$. However, the maximum receiver degree in this network is $\Delta_R=3$, hence random Gaussian coding and interference decoding can only achieve the symmetric DoF of $\frac{1}{3}$ which is less than the outer bound.

On the other hand, it is obvious that the maximal independent sets of the conflict graph of this network, shown in Figure \ref{fig4b}, are $\{1,2\},\{2,3\},\{3,4\},\{4,5\}$ and $\{5,1\}$. By assigning one time slot to each of these sets, we can achieve the symmetric DoF of $\frac{2}{5}$ because each user is repeated twice in these sets, therefore meeting the outer bound of $\frac{2}{5}$ mentioned earlier. A corresponding 5:2-coloring is also shown in Figure \ref{fig4c}. Hence, in this example, interference avoidance outperforms random Gaussian coding and interference decoding.
\demo
\end{ex}

Taking a closer look at the two schemes presented in this section, they can be viewed as two extremes of a spectrum. Random Gaussian coding and interference decoding tries to decode all the interference at all the receivers by adopting a random code which does not make efficient use of the topology of the network. On the other side, interference avoidance tries to prevent the mutually interfering nodes from transmitting at the same time, which causes no interference to occur at the receivers. Therefore, one may think of using a scheme that is naturally between these two extremes; i.e. using some kind of structured code that makes best use of the topology of the network and does not necessarily try to avoid the interference at the receivers, but at the same time enables the receivers to decode their desired messages. This leads to a new scheme which will be introduced in the following section.

\subsection{Structured Repetition Coding}

We now present a scheme based on structured repetition codes at the transmitters so that we can better exploit structure of network topology. This scheme unifies the two schemes presented in Section \ref{bench} in the way that it not only enables the receivers to decode their intended symbols without necessarily decoding all the interference, but it also allows mutually interfering users to possibly send data at the same time, implying that the scheme can potentially outperform both benchmark schemes presented in Section \ref{bench}. We will motivate the idea of structured repetition coding through the following example. Before starting the example, we need the following definition.

\begin{defi}\label{def_match}
For a graph $G=(\mathcal{V},\mathcal{E})$, a \emph{matching} is a subset of edges no two of which share a common vertex. The \emph{matching number} of $G$, denoted by $\mu(G)$, is the size of a maximum matching of $G$ (a matching of $G$ containing the largest possible number of edges).
\trig
\end{defi}

\begin{ex}\label{ex5}

Consider the 6-user network in Figure \ref{fig7}. We claim that in this network, the symmetric DoF of $\frac{1}{3}$ is achievable, while the benchmark schemes discussed in the previous section can at most achieve a symmetric DoF of $\frac{1}{4}$.
\begin{figure}[h]
\centering
\includegraphics[trim = 2in 3in 2in 3.5in, clip,width=0.3\textwidth]{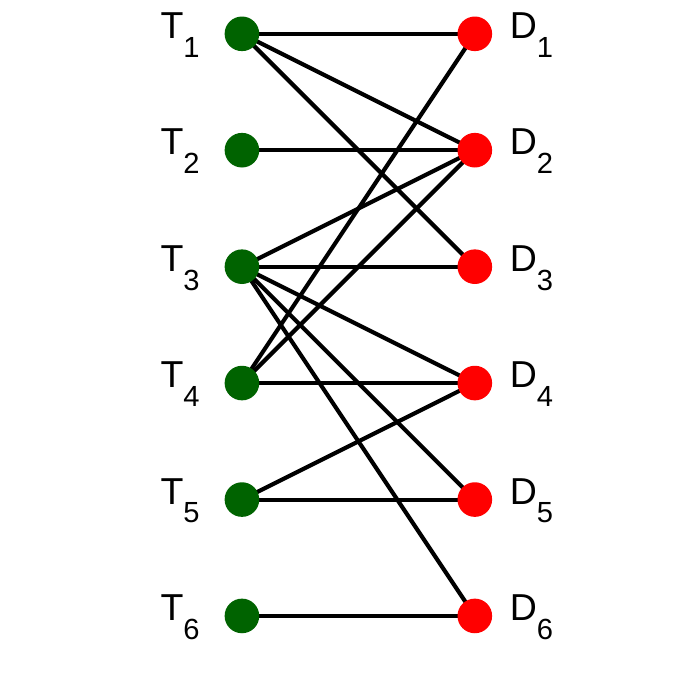}
\caption{A 6-user interference network in which random Gaussian coding and interference avoidance are suboptimal.}
\label{fig7}
\end{figure}

In this network, the outer bound for the symmetric DoF is $d_{sym}\leq\frac{1}{3}$ by Corollary \ref{cor1}, because the sets $\mathcal{A}=\{2\}$ and $\mathcal{S}=\{1,2,3\}$ satisfy the conditions of the corollary; i.e. in the subnetwork $\mathcal{S}$, we can generate statistically similar versions of the signals at receivers 3 and 1 by having the received signal at receiver 2. Therefore, $d_{sym}\leq \frac{|\mathcal{A}|}{|\mathcal{S}|}=\frac{1}{3}$.

However, in terms of the achievable schemes, Theorem \ref{th3} indicates that the best symmetric DoF achievable by random Gaussian coding and interference decoding is $\frac{1}{\Delta_R}=\frac{1}{4}$. Also, the maximal independent sets of this network are $\{1,5,6\}$, $\{2,5,6\}$, $\{3\}$ and $\{4,6\}$. Therefore, Theorem \ref{th4} states that the maximum symmetric DoF which interference avoidance can achieve is $\frac{1}{4}$. Thus, our two previous schemes both achieve the same symmetric DoF of $\frac{1}{4}$ which is strictly lower than the outer bound of $\frac{1}{3}$. Now, let us see if the achievable symmetric DoF can be improved.

Targeting the symmetric DoF of $\frac{1}{3}$, we can think of an achievable scheme in which each transmitter has one symbol to be sent within three time slots such that all the receivers can decode their desired messages. To this end, we create a \emph{transmission matrix} $\mathbf{T}\in\{0,1\}^{6\times 3}$ where $\mathbf{T}_{ik}=1$ if transmitter $i$ sends its single symbol $X_i$ in time slot $k$ and $\mathbf{T}_{ik}=0$ if transmitter $i$ is silent in time slot $k$. Consider the following matrix.
\begin{equation}\label{eq19}
\mathbf{T}=\begin{pmatrix}
1 & 0 & 1 & 1 & 1 & 0\\
0 & 1 & 0 & 0 & 1 & 1\\
0 & 0 & 1 & 1 & 0 & 0
\end{pmatrix}^T.
\end{equation}

As mentioned above, the first row of (\ref{eq19}) means that transmitter 1 sends its only symbol $X_1$ in time slot 1 and remains silent otherwise, the second row means that transmitter 2 sends its symbol $X_2$ in time slot 2, the third row implies that transmitter 3 repeats its symbol $X_3$ in time slots 1 and 3, etc. We will now show that with this transmission matrix, all the receivers can create interference-free versions of their desired symbols for almost all values of channel gains. As an example, let us focus on receiver 4. The signals that $\text{D}_4$ receives in three time slots are as follows.
\begin{align*}
Y_4[1]&=g_{34}[1]X_3+g_{44}[1]X_4+g_{54}[1]X_5+Z_4[1]\\
Y_4[2]&=g_{54}[2]X_5+Z_4[2]\\
Y_4[3]&=g_{34}[3]X_3+g_{44}[3]X_4+Z_4[3].
\end{align*}

Since $\text{D}_4$ is aware of the channel gains of all the links connected to it at all times, it can create the following signal.
\begin{align*}
Y'_4[1,2]:=Y_4[1]-\frac{g_{54}[1]}{g_{54}[2]}Y_4[2]=g_{34}[1]X_3+g_{44}[1]X_4+Z'_4[1,2],
\end{align*}
where $Z'_4[1,2]$ is a noise term with bounded variance. Now, it is clear that from $Y_4[3]$ and $Y'_4[1,2]$, $\text{D}_4$ can create an interference-free version of $X_4$ as follows.
\begin{align*}
\frac{g_{34}[1]Y_4[3]-g_{34}[3]Y'_4[1,2]}{g_{34}[1]g_{44}[3]-g_{34}[3]g_{44}[1]}=X_4+\tilde{Z}_4,
\end{align*}
where $\tilde{Z}_4$ has a bounded variance. The above combination of the signals is possible if \linebreak $g_{34}[1]g_{44}[3]-g_{34}[3]g_{44}[1]\neq 0$ which holds for almost all values of channel gains, because the channel gains are i.i.d. and drawn from continuous distributions.

The fact that for almost all values of the channel gains, there exists a linear combination of the received signals at receiver 4 which is an interference-free version of $X_4$ can also be viewed in terms of the matching number of a bipartite graph. The idea is to first create an ``effective'' transmission matrix $\bar{\mathbf{T}}^4$ for receiver 4, which is defined as a $6\times 3$ matrix, where $\bar{\mathbf{T}}_{ik}^4=\mathbf{M}_{i4}\mathbf{T}_{ik},\:\forall i\in[1:6], k\in[1:3]$, as shown in (\ref{Tbar}). In words, $\bar{\mathbf{T}}^4$ is the same as $\mathbf{T}$ with the distinction that the rows corresponding to the transmitters which are not connected to $\text{D}_4$ are set to zero.
\begin{align}\label{Tbar}
\bar{\mathbf{T}}^4=\begin{pmatrix}
0 & 0 & 1 & 1 & 1 & 0\\
0 & 0 & 0 & 0 & 1 & 0\\
0 & 0 & 1 & 1 & 0 & 0
\end{pmatrix}^T.
\end{align}

This matrix corresponds to a bipartite graph $\bar{G}^4$, shown in Figure \ref{Gbar}, with the set of vertices $\{v_1,...,v_6\}\cup\{v'_1,v'_2,v'_3\}$, where $v_i$ is connected to $v'_k$ iff $\bar{\mathbf{T}}_{ik}^1=1,\:\forall i\in[1:6], k\in[1:3]$.
\begin{figure}[hbt]
\centering
\begin{subfigure}[b]{0.45\textwidth}
\centering
\includegraphics[trim = 1in 2.45in .8in 2.5in, clip,width=.8\textwidth]{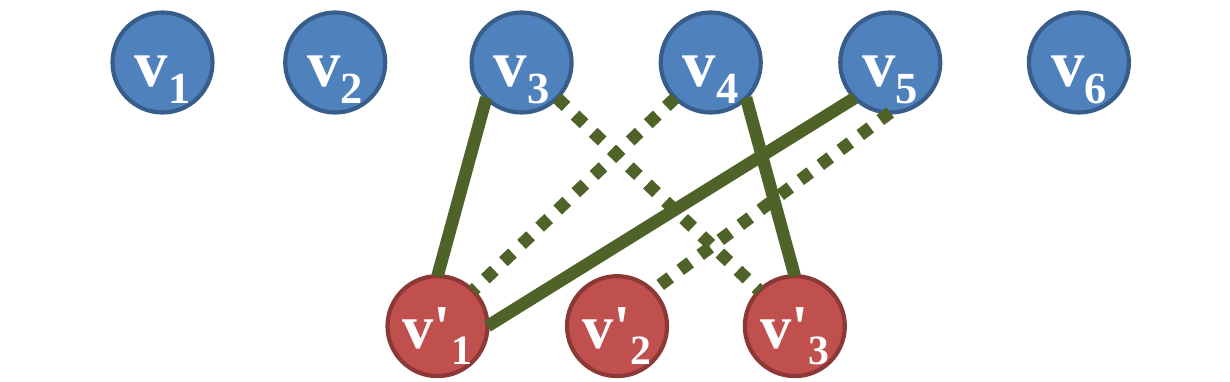}
\caption{}
\label{Gbar}
\end{subfigure}
~
\begin{subfigure}[b]{0.45\textwidth}
\centering
\includegraphics[trim = 1in 2.45in .8in 2.5in, clip,width=.8\textwidth]{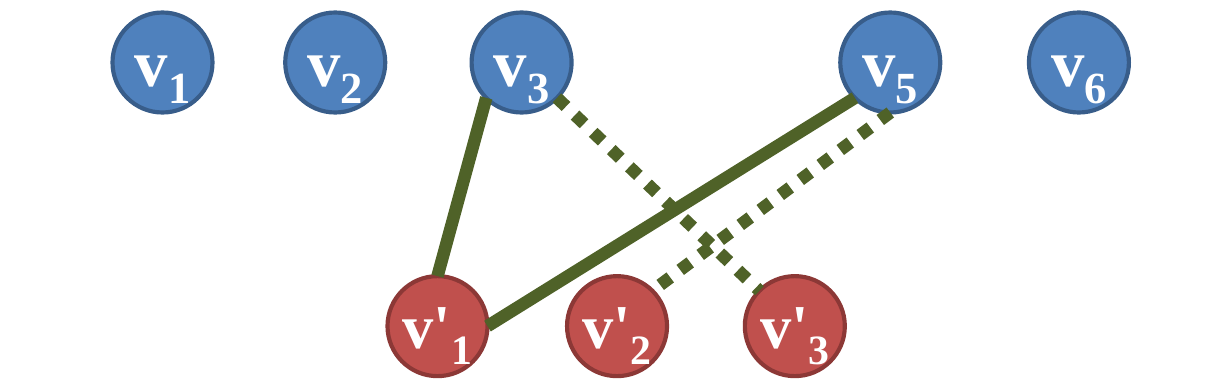}
\caption{}
\label{Gbar2}
\end{subfigure}
\caption{(a) The bipartite graph $\bar{G}^4$ corresponding to the matrix $\bar{\mathbf{T}}^4$ in (\ref{Tbar}), and (b) the graph $\bar{G}^4\setminus4$, which is the same as $\bar{G}^4$ after removing $v_4$ and its corresponding edges. In both graphs, the dashed edges correspond to a maximum matching.}
\end{figure}

Note that the matching number of $\bar{G}^4$, denoted by $\mu(\bar{G}^4)$, is equal to 3 and a maximum matching of $\bar{G}^4$ is shown in Figure \ref{Gbar}. However, as shown in Figure \ref{Gbar2}, upon removal of $v_4$ and its corresponding edges from $\bar{G}^4$, the matching number reduces to 2. As we show in Lemma \ref{match}, this reduction in the matching number is equivalent to the fact that for almost all values of the channel gains, there exists a linear combination of the signals at receiver 4 which is an interference-free version of $X_4$. Theorem \ref{th5} shows that this procedure reduces the problem of checking whether the transmission matrix $\mathbf{T}$ is successful or not to a bipartite matching problem.

Therefore, user 4 can achieve $\frac{1}{3}$ degrees-of-freedom. Arguments similar to the one above can show that all the other receivers can create interference-free versions of their desired symbols either, by linearly combining their received signals in three time slots. In particular, $\text{D}_1$ needs to combine its received signals at time slots 1 and 3, whereas $\text{D}_2$, $\text{D}_3$, $\text{D}_5$ and $\text{D}_6$ only need their received signals at time slots 2, 3, 2 and 2, respectively. Therefore, this scheme, which we will call \emph{structured repetition coding}, can achieve the symmetric DoF of $\frac{1}{3}$. This inner bound meets the outer bound, indicating that structured repetition coding is optimal in the network of Figure \ref{fig7}, contrary to the two benchmark schemes which perform suboptimally in this example.
\demo
\end{ex}

Motivated by Example \ref{ex5}, we now formally define structured repetition coding. In what follows, for a general matrix $\mathbf{T}$, we use $\mathbf{T}_{l,*}$ to denote the $l^{th}$ row of $\mathbf{T}$.

\begin{defi}\label{d4}
Consider a $K$-user interference network with adjacency matrix $\mathbf{M}$. Also, consider a matrix $\mathbf{T}\in\{0,1\}^{mK \times n}$, for some $m,n \in \mathbb{N}$, satisfying
\begin{align}\label{eq_tr}
\sum_{l=(i-1)m+1}^{im}\mathbf{T}_{lk}\leq1,\:\forall k\in[1:n],\:\forall i\in[1:K],
\end{align}
which, in words, means that there exists at most a single 1 in positions $(i-1)m+1$ to $im$ of each column $k$, for all $i\in[1:K]$ and for all $k\in[1:n]$. Then, \emph{structured repetition coding with transmission matrix $\mathbf{T}$} is defined as a scheme, in which transmitter $\text{T}_i$ ($i \in [1:K]$) intends to deliver $m$ independent symbols, denoted by $\left\lbrace\tilde{X}_l\right\rbrace_{l=(i-1)m+1}^{im}$, to receiver $\text{D}_i$ in $n$ time slots, using the following encoding and decoding procedure.

\begin{itemize}
\item Transmitter $\text{T}_i$ ($i\in[1:K]$) creates its transmit vector, denoted by $X_i^n$, as follows.
\begin{align*}
X_i^n=\sum_{l=(i-1)m+1}^{im} \mathbf{T}_{l,*}^T \tilde{X}_l.
\end{align*}

In words, this means that at each time slot $k$, transmitter $i$ ($i \in [1:K]$) looks for index $l \in [(i-1)m+1: im]$ such that $\mathbf{T}_{lk}=1$ (note that due to (\ref{eq_tr}), there is at most one such $l$) and transmits $\tilde{X}_l$ in that time slot (if such an index cannot be found, the transmitter will remain silent).

\item At the end of the transmission, receiver $\text{D}_j$ ($j\in[1:K]$) receives
\begin{align*}
Y_j^n&=\sum_{i=1}^K \mathbf{M}_{ij} g_{ij}^n X_{i}^n+Z_j^n\\
&=\sum_{i=1}^K \mathbf{M}_{ij} g_{ij}^n \left(\sum_{l=(i-1)m+1}^{im} \mathbf{T}_{l,*}^T \tilde{X}_l\right) +Z_j^n\\
&=\sum_{l=1}^{mK} \mathbf{M}_{\lceil\frac{l}{m} \rceil j} g_{\lceil\frac{l}{m} \rceil j}^n \mathbf{T}_{l,*}^T \tilde{X}_l +Z_j^n.
\end{align*}

Then, $\text{D}_j$ looks for vectors $\mathbf{u}_l\in\mathbb{C}^n,\:l\in[(j-1)m+1:jm]$ such that
\begin{align}\label{eq_decod_matrix}
\mathbf{G}^j \mathbf{u}_l=\mathbf{I}_l^{mK},\:\forall l\in[(j-1)m+1:jm],
\end{align}
where $\mathbf{G}^j$ is the $mK\times n$ matrix whose $lk^{th}$ element is defined as $\mathbf{G}_{lk}^j=\mathbf{M}_{\lceil\frac{l}{m} \rceil j} g_{\lceil\frac{l}{m} \rceil j}[k] \mathbf{T}_{lk}$, and if it can find such $\mathbf{u}_l$'s, it will reconstruct a noisy, but interference-free, version of each symbol $\tilde{X}_l$ by projecting $Y_j^n$ along the direction of $\mathbf{u}_l$, i.e.
\begin{align*}
(Y_j^n)^T \mathbf{u}_l=\tilde{X}_l+(Z_j^n)^T \mathbf{u}_l,\:\forall l\in[(j-1)m+1:jm].
\end{align*}
\end{itemize}
\trig
\end{defi}

\begin{remk}
If the conditions in (\ref{eq_decod_matrix}) are satisfied, then by using an outer code for each of the symbols $\tilde{X}_l$, $l\in[1:mK]$, a rate of $C_l=\mathbb{E}\left[\log\left(1+\frac{P}{\|\mathbf{u}_l\|_2^2}\right)\right]\geq \log(P)-\mathbb{E}\left[\log\left({\|\mathbf{u}_l\|_2^2}\right)\right]$ over each symbol can be achieved, where the expectation is taken with respect to the channel gain values. Since $\mathbb{E}\left[\log\left({\|\mathbf{u}_l\|_2^2}\right)\right]$ does not scale with the transmit power $P$, and as shown in Appendix \ref{apx4}, its value is finite, the scheme guarantees 1 DoF per symbol.
\end{remk}

In the remainder of this section, we will address the conditions that the transmission matrix $\mathbf{T}$ needs to satisfy in order to guarantee the existence of $\mathbf{u}_l$'s satisfying (\ref{eq_decod_matrix}), hence being able to neutralize the interference at all the receivers. We will then use these conditions to characterize the symmetric DoF that is achievable by structured repetition coding.

\begin{defi}\label{d5}
Consider a $K$-user interference network with adjacency matrix $\mathbf{M}$ and structured repetition coding with transmission matrix $\mathbf{T}\in\{0,1\}^{mK\times n}$. For each $j\in[1:K]$, $\bar{\mathbf{T}}^j$ is an $mK\times n$ matrix whose $lk^{th}$ element is defined as
\begin{align*}
\bar{\mathbf{T}}_{lk}^j=\mathbf{T}_{lk} \mathbf{M}_{\lceil\frac{l}{m}\rceil j},\:\forall l\in[1:mK],\:\forall k\in[1:n].
\end{align*}

Moreover, $\bar{G}^j$ is defined as the bipartite graph with the set of vertices $\mathcal{V}=\{v_1,...,v_{mK}\}\cup\{v'_1,...,v'_n\}$ whose adjacency matrix is $\bar{\mathbf{T}}^j$; i.e., for any $l\in[1:mK]$ and $k\in[1:n]$, $v_l$ is connected to $v'_k$ if and only if $\bar{\mathbf{T}}_{lk}^j=1$. Also, for any $l\in[1:mK]$, we use the notation $\bar{G}^j\setminus l$ to denote the subgraph of $\bar{G}^j$ with node $v_{l}$ and its incident edges removed.
\trig
\end{defi}

The above definitions make us ready to state our theorem about the graph theoretic conditions that a transmission matrix $\mathbf{T}$ needs to satisfy to achieve a symmetric DoF of $\frac{m}{n}$.


\begin{theo}\label{th5}
Consider a $K$-user interference network with adjacency matrix $\mathbf{M}$. If a transmission matrix $\mathbf{T}\in\{0,1\}^{mK\times n}$ satisfies the following conditions
\begin{align*}
\mu(\bar{G}^j)-\mu(\bar{G}^j\setminus l)=1,\:\forall l\in[(j-1)m+1:jm],\:\forall j\in[1:K],
\end{align*}
where $\bar{G}^j$ and $\bar{G}^j\setminus l$ are defined in Definition \ref{d5}, then for almost all values of channel gains, there exist vectors $\left\lbrace \mathbf{u}_l\right\rbrace_{l=1}^{mK}$ satisfying
\begin{align*}
(Y_j^n)^T \mathbf{u}_l=\tilde{X}_l+(Z_j^n)^T \mathbf{u}_l,\:\forall l\in[(j-1)m+1:jm],\:\forall j\in[1:K],
\end{align*}
where $Y_j^n$ and $\tilde{X}_l$ are defined in Definition \ref{d4}. Hence, structured repetition coding with transmission matrix $\mathbf{T}$ achieves the symmetric DoF of $\frac{m}{n}$.
\end{theo}

Theorem \ref{th5} immediately leads to the following corollary.

\begin{corl}\label{corl2}
Consider a $K$-user interference network with adjacency matrix $\mathbf{M}$. Then, the following symmetric DoF is achievable by structured repetition coding (see Definition \ref{d4}).
\begin{align*}
\underset{n\in\mathbb{N}}{\emph{sup}}\:\underset{m\in[1:n]}{\emph{max}}\: &\frac{m}{n}\\
\emph{s.t.}\:&\exists \mathbf{T}\in\{0,1\}^{mK\times n}:\mu(\bar{G}^j)-\mu(\bar{G}^j\setminus l)=1,\:\forall l\in[(j-1)m+1:jm],\:\forall j\in[1:K],
\end{align*}
where $\bar{G}^j$ and $\bar{G}^j\setminus l$ are defined in Definition \ref{d5}.
\end{corl}

\begin{remk}
While in the optimization problem of Corollary \ref{corl2}, the value of $\underset{m\in[1:n]}{\max}\: \frac{m}{n}$ is optimized over $n\in\mathbb{N}$, we will limit the range space for $n$ to be bounded as $n\in[1:K+1]$ in order to numerically evaluate the inner bounds in Section \ref{sim}, and as we will see, the inner bounds derived after this reduction match the outer bounds in most of the topologies. This reduces the optimization problem in Corollary \ref{corl2} to a combinatorial optimization problem that can be solved for relatively small networks. Finding efficient algorithms to solve it for general networks is an interesting open problem.
\end{remk}

\begin{remk}
The structured repetition coding scheme illustrates the fact that even in the case where the channel gains change i.i.d. over time (i.e. coherence time of 1 time slot), it is possible to exploit network topology in order to design a carefully-chosen repetition pattern at the transmitters which enables the receivers to neutralize all the interference. However, as the coherence time of the channel increases, there would be other opportunities that can be utilized, such as aligning the interference, as in \cite{jafar_old,jafar}.
\end{remk}

The existence of a vector $\mathbf{u}_l$ satisfying the conditions in Theorem \ref{th5} is equivalent to the existence of a vector $\mathbf{u}_l$ satisfying the conditions in (\ref{eq_decod_matrix}), i.e. $\mathbf{G}^j\mathbf{u}_l=\mathbf{I}_l^{mK}$, where $\mathbf{G}^j$ is an $mK\times n$ matrix whose entries are either zero or i.i.d. random variables (corresponding to the channel gains $g_{ij}$). This enables us to use the following lemma, proved in Appendix C, which addresses the existence of $\mathbf{u}_l$'s satisfying $\mathbf{G}^j\mathbf{u}_l=\mathbf{I}_l^{mK}$ for such structured random matrices $\mathbf{G}^j$.

\begin{lem}\label{match}
Consider a bipartite graph $G=(\{v_1,...,v_m\}\cup\{v'_1,...,v'_n\},\mathcal{E})$ with a corresponding $m\times n$ adjacency matrix $\mathbf{T}$ where $\mathbf{T}_{ij}=1$ if $v_i$ is connected to $v'_j$ and $\mathbf{T}_{ij}=0$ otherwise. Also, define $\tilde{\mathbf{T}}$ to be an $m\times n$ matrix for which $\tilde{\mathbf{T}}_{ij}=g_{ij}\mathbf{T}_{ij}$, where $g_{ij}$'s are i.i.d. random variables drawn from a continuous distribution. If for some $l\in[1:m]$, $\mu(G)-\mu(G\setminus l)=1$ (where $G\setminus l$ denotes the subgraph of $G$ with node $v_{l}$ and its incident edges removed), then for almost all values of $g_{ij}$'s, there exists a vector $\mathbf{u}\in\mathbb{C}^n$ such that
\begin{align*}
\tilde{\mathbf{T}}\mathbf{u}=\mathbf{I}_l^m,
\end{align*}
where $\mathbf{I}_l^m$ is the $l^{th}$ column of the $m\times m$ identity matrix. Moreover, $\|\mathbf{u}\|_2=\left\|(\tilde{\mathbf{T}}^l)^{-1}\mathbf{I}_l^{\mu(G)}\right\|_2$, where $\tilde{\mathbf{T}}^l$ is a $\mu(G)\times\mu(G)$ submatrix of $\tilde{\mathbf{T}}$ corresponding to a maximum matching in $G$.
\end{lem}

\begin{proof}[Proof of Theorem \ref{th5}]
Following Definitions \ref{d4} and \ref{d5}, the received vector of receiver $j$ ($j\in[1:K]$) can be written as
\begin{align*}
Y_j^n&=\sum_{l=1}^{mK} \mathbf{M}_{\lceil\frac{l}{m} \rceil j} g_{\lceil\frac{l}{m} \rceil j}^n \mathbf{T}_{l,*}^T \tilde{X}_l +Z_j^n\\
&=\sum_{l=1}^{mK} g_{\lceil\frac{l}{m} \rceil j}^n (\bar{\mathbf{T}}_{l,*}^j)^T \tilde{X}_l +Z_j^n,
\end{align*}
and it needs vectors $\left\lbrace \mathbf{u}_l\right\rbrace_{l=(j-1)m+1}^{jm}$ such that
\begin{align}\label{eq_comb}
(Y_j^n)^T \mathbf{u}_l=\tilde{X}_l+(Z_j^n)^T \mathbf{u}_l,\:\forall l\in[(j-1)m+1:jm].
\end{align}

This means that for almost all values of the channel gains, there must exist vectors $\left\lbrace \mathbf{u}_l\right\rbrace_{l=(j-1)m+1}^{jm}$ satisfying
\begin{align*}
\mathbf{G}^j\mathbf{u}_l=\mathbf{I}_l^{mK},\:\forall l\in[(j-1)m+1:jm],
\end{align*}
where $\mathbf{I}_l^{mK}$ is the $l^{th}$ column of the $mK\times mK$ identity matrix, and $\mathbf{G}^j$ is an $mK\times n$ matrix whose $lk^{th}$ element is defined as $\mathbf{G}^j_{lk}=g_{\lceil\frac{l}{m}\rceil j}[k]\bar{\mathbf{T}}_{lk}^j$. Due to the specific structure of the transmission matrix $\mathbf{T}$ described in (\ref{eq_tr}), $\mathbf{G}^j$ has i.i.d. random entries and zeros wherever $\bar{\mathbf{T}}^j$ has ones and zeros, respectively. This enables us to make use of Lemma \ref{match}, therefore proving the existence of vectors $\mathbf{u}_l$, $\forall l\in[(j-1)m+1:jm] $, $\forall j\in[1:K]$.


The only remaining issue to address is the noise variance in (\ref{eq_comb}). The capacity of the channel in (\ref{eq_comb}) is equal to 
\begin{align}\label{capc}
C_l=\mathbb{E}\left[\log\left(1+\frac{P}{\|\mathbf{u}_l\|_2^2}\right)\right],
\end{align}
where the expectation is taken with respect to the channel gains. Lemma \ref{match} implies that $\|\mathbf{u}_l\|_2=\left\|\left(\tilde{\mathbf{G}}^{j,l}\right)^{-1} \mathbf{I}_l^{\mu(\bar{G}^j)}\right\|_2$, where $\tilde{\mathbf{G}}^{j,l}$ is a $\mu(\bar{G}^j)\times\mu(\bar{G}^j)$ submatrix of $\mathbf{G}^j$ corresponding to a maximum matching in $\bar{G}^j$. Combining this with (\ref{capc}), we can write
\begin{align*}
C_l&\geq\mathbb{E}\left[\log\left(\frac{P}{\|\mathbf{u}_l\|_2^2}\right)\right]\\
&= \log(P)-\mathbb{E}\left[\log\left({\|\mathbf{u}_l\|_2^2}\right)\right]\\
&= \log(P)-\mathbb{E}\left[\log\left(\left\|\left(\tilde{\mathbf{G}}^{j,l}\right)^{-1} \mathbf{I}_l^{\mu(\bar{G}^j)}\right\|^2_2\right)\right].\numberthis\label{eq_noise}
\end{align*}

Now, note that
$\mathbb{E}\left[\log\left(\left\|\left(\tilde{\mathbf{G}}^{j,l}\right)^{-1} \mathbf{I}_l^{\mu(\bar{G}^j)}\right\|^2_2\right)\right]$ does not scale with the transmit power $P$ and as we show in Appendix \ref{apx4}, its value is finite. Therefore, the outer code on each of the symbols $\tilde{X}_l$ guarantees 1 degree-of-freedom to be achieved for that symbol.

Hence, if all the conditions of the theorem are satisfied, then all the receivers can create interference-free versions of their $m$ desired symbols, implying that structured repetition coding with transmission matrix $\mathbf{T}$ can achieve the symmetric DoF of $\frac{m}{n}$.
\end{proof}

Theorem \ref{th5}, therefore, implies that for any given network topology, it suffices to carefully choose a well-structured transmission matrix $\mathbf{T}\in\{0,1\}^{mK\times n}$ which satisfies the graph theoretic conditions mentioned in the theorem. This makes the symmetric DoF of $\frac{m}{n}$ achievable through structured repetition coding.

\section{Numerical Analysis}\label{sim}

In this section, we will evaluate our inner and outer bounds for two diverse classes of network topologies. We will examine the possible network instances in two scenarios of 6-user networks with 6 square cells and 6-user networks with 1 central and 5 surrounding base stations. The goal is to study the tightness of our inner and outer bounds, compare the performance of the achievable schemes presented in Section \ref{ach}, and study the effect of network density on the fraction of topologies in which structured repetition coding outperforms benchmark schemes. Note that for the structured repetition coding scheme, we search over all transmission matrices $\mathbf{T}\in\{0,1\}^{mK\times n}$ for which $n\leq K+1$, due to computational tractability. We seek to answer the following questions.

\begin{itemize}
\item Do there exist any network topologies in which our inner and outer bounds on the symmetric DoF do not meet? On the other hand, among the networks in which the bounds are tight, what are the possible values of the symmetric DoF and how are these values distributed?

\item Focusing on the topologies in which the inner and outer bounds meet, what is the impact of the sparsity or density of the network graph on the gains that can be obtained beyond the benchmark schemes by using only the knowledge about network topology?

\item  What is the fraction of the topologies in which structured repetition coding can outperform the benchmark schemes? How much can the sole knowledge about network topology help to go beyond random Gaussian coding and interference avoidance?

\end{itemize}

We will address these questions in the following classes of networks.

\subsection{6-User Networks with 6 Square Cells}\label{6cell}

The networks that we consider in this section are represented by 6 square cells, each one having a base station $\text{BS}_i$ in the center, $i\in[1:6]$, with a mobile user inside the cell. An example can be seen in Figure \ref{fig8}.
\begin{figure}[h]
\centering
\includegraphics[trim = 1.2in 3in .9in 3in, clip,width=0.35\textwidth]{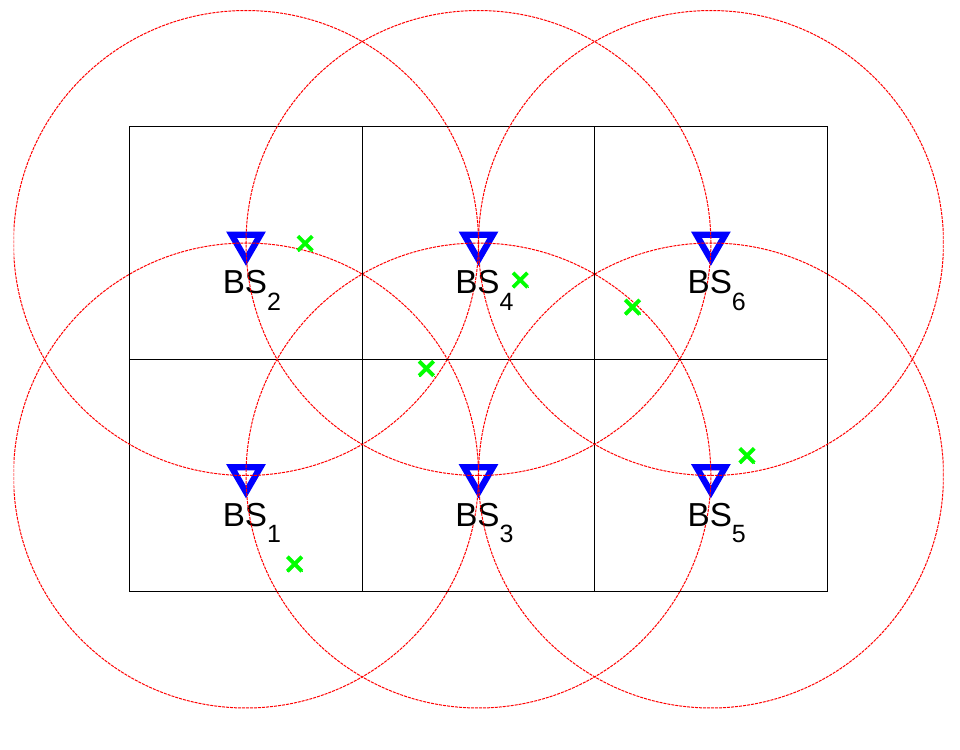}
\caption{A 6-cell network realization where the blue triangles, green crosses, black squares and red circles represent base stations, mobile users, cell boundaries and coverage area of base stations, respectively.}
\label{fig8}
\end{figure}

In this figure, the blue triangles represent base stations, the green crosses represent mobile users, the black squares represent the cells and the red circles depict the coverage area of each base station. It is obvious that any placement of the mobile users corresponds to a partially-connected 6-user interference network.

In what follows, we will generalize this model to all possible topologies in which a mobile user in a cell can receive interference from any \emph{nonempty subset} of its three adjacent BS's, together with the signal from its own BS. For instance in Figure \ref{fig8}, user 2 can receive interference from any nonempty subset of $\{\text{BS}_1,\text{BS}_3,\text{BS}_4\}$ and user 4 can receive interference from any nonempty subset of $\{\text{BS}_1,\text{BS}_2,\text{BS}_3\}$ or $\{\text{BS}_3,\text{BS}_5,\text{BS}_6\}$ (corresponding to left and right halves of the cell, respectively). This implies that the degree of each receiver is no less than 2 and no more than 4. Ignoring isomorphic topologies, there are in total 22,336 unique topologies in this class. For each of these topologies, we evaluated our inner and outer bounds to draw the following conclusions.

\begin{enumerate}
\item We note that quite interestingly, our bounds are tight for all cases, except for 16 distinct topologies. We will discuss two of these 16 topologies in Section \ref{conc}. For the remaining networks, which we will hereby focus on, the gap is zero, implying that our bounds determine the symmetric DoF for most networks in this class. In these networks, the symmetric DoF only takes 4 distinct values in $\lbrace\frac{1}{4},\frac{1}{3},\frac{2}{5},\frac{1}{2}\rbrace$ with the distribution shown in Figure \ref{fig9}.

\begin{figure}[hbt]
\centering
\includegraphics[trim = 1.65in 3.3in 1.75in 3.4in, clip,width=0.35\textwidth]{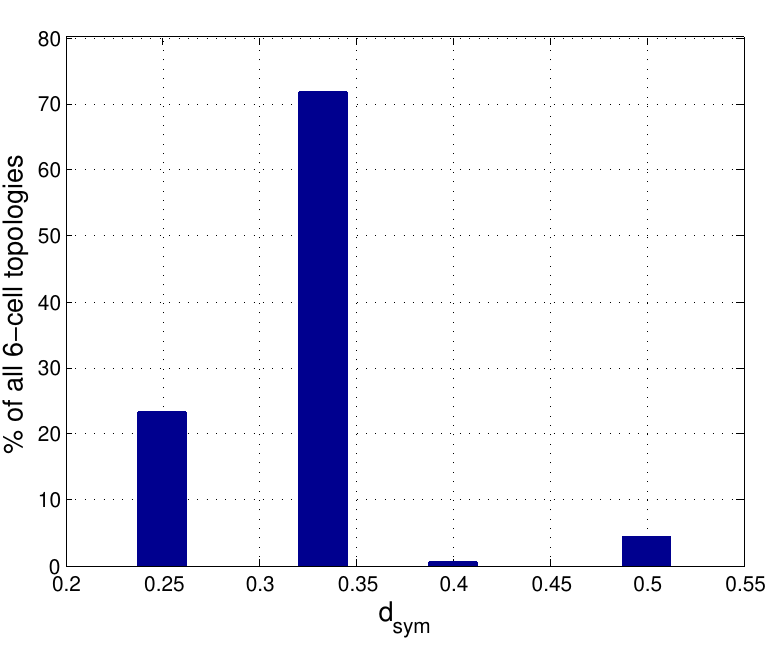}
\caption{Distribution of $d_{sym}$ among 6-cell networks in which our bounds are tight.}
\label{fig9}
\end{figure}

\item Figure \ref{fig10} illustrates the impact of the number of interfering links on the performance of structured repetition coding compared to benchmark schemes. As it is clear, the gain is not much when the network is too dense. However, if the density of the network, characterized by the number of cross links in the network, is at a moderate level, then the gain of structured repetition coding over the benchmark schemes can be significant. It is worth mentioning that there are totally around 50 percent and 10 percent of the networks in which structured repetition coding outperforms random Gaussian coding and interference avoidance, respectively. Moreover, structured repetition coding outperforms \emph{both} benchmark schemes in 1167 network topologies, which constitute more than 5 percent of all the networks. This means that even with a sole knowledge of network topology, one can perform better than both of the benchmark schemes.
\begin{figure}[h]
\centering
\includegraphics[trim = 0in 2.5in 0in 2.6in, clip,width=0.52\textwidth]{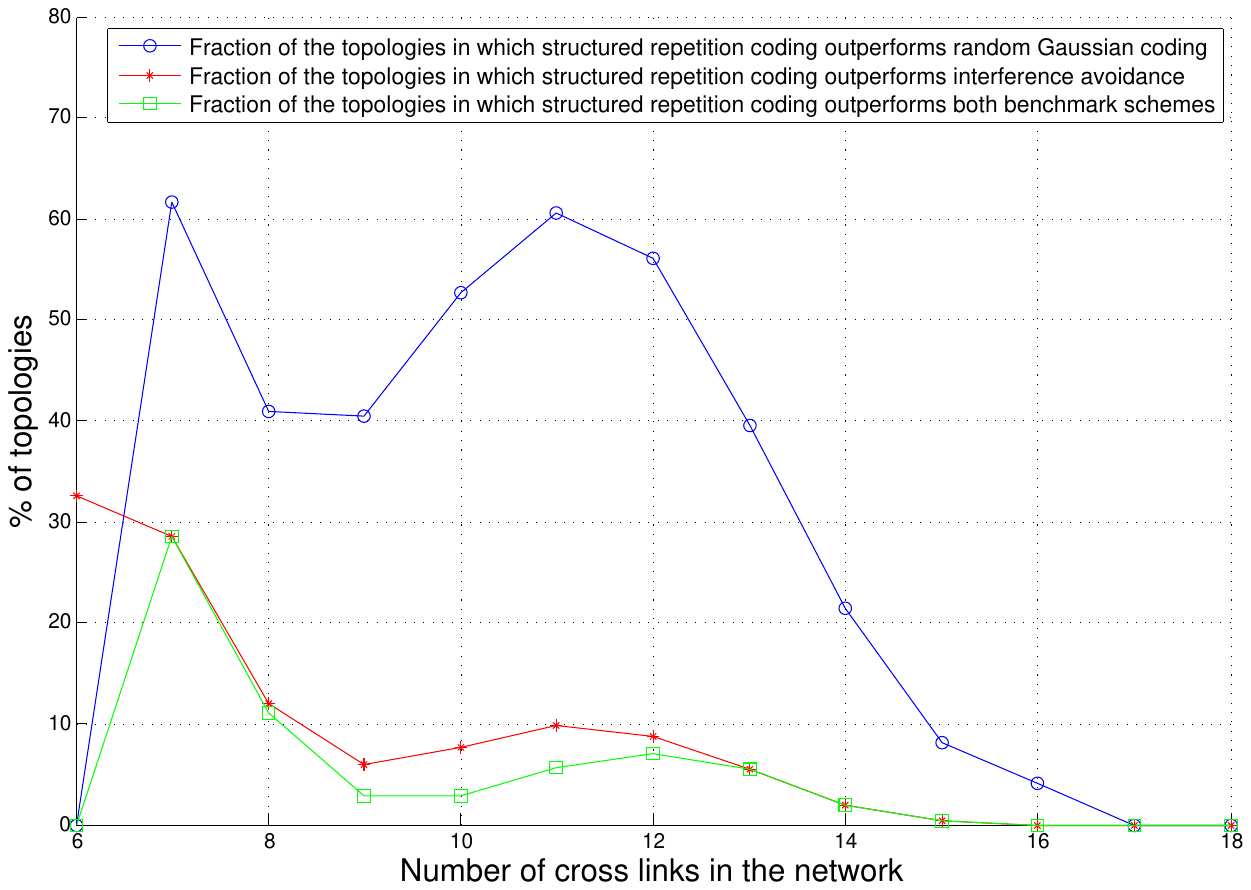}
\caption{Effect of network density on the fraction of networks in which structured repetition coding outperforms benchmark schemes in 6-user cellular networks.}
\label{fig10}
\end{figure}

\item Turning our focus to the networks where structured repetition coding outperforms the benchmark schemes, it is interesting to know the value of the gains obtained over them. Among the networks in which structured repetition coding outperforms random Gaussian coding, the gain of the former scheme over the latter takes 5 distinct values in $\lbrace \frac{6}{5},\frac{4}{3},\frac{3}{2},\frac{8}{5},2\rbrace$, distributed as shown in Figure \ref{fig11a}. Also, Figure \ref{fig11b} illustrates the distribution of the gain of structured repetition coding over interference avoidance among the networks in which this gain is greater than unity. This gain can take 4 distinct values in $\lbrace \frac{6}{5},\frac{5}{4},\frac{4}{3},\frac{3}{2}\rbrace$.
\begin{figure}
\centering
\begin{subfigure}[b]{0.47\textwidth}
\centering
\includegraphics[trim = 1.65in 3.3in 1.75in 3.4in, clip,width=0.8\textwidth]{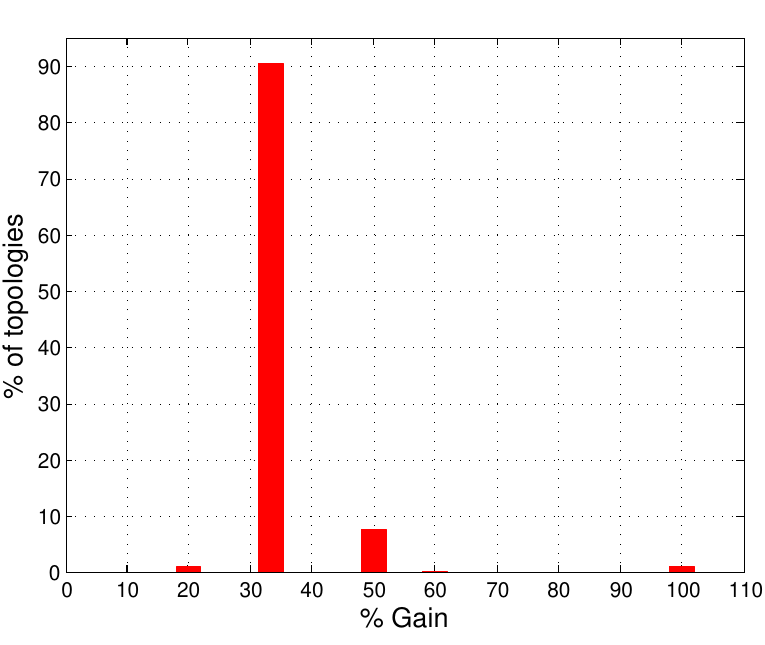}
\caption{Distribution of the gain of structured repetition coding over $\frac{1}{\Delta_{R}}$ (random Gaussian coding).}
\label{fig11a}
\end{subfigure}
~
\begin{subfigure}[b]{0.47\textwidth}
\centering
\includegraphics[trim = 1.6in 3.3in 1.75in 3.4in, clip,width=0.8\textwidth]{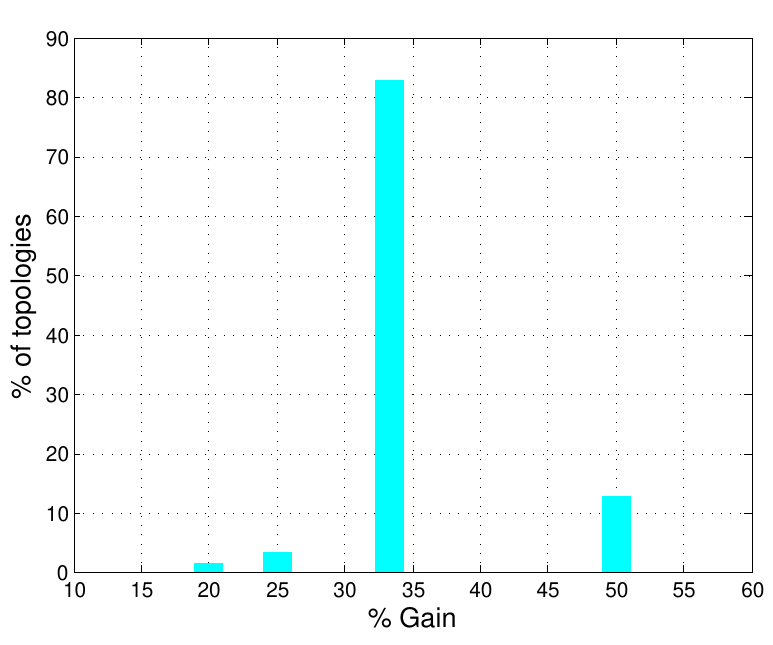}
\caption{Distribution of the gain of structured repetition coding over interference avoidance.}
\label{fig11b}
\end{subfigure}
\caption{Comparison of achievable schemes in 6-user cellular networks.}
\label{fig11}
\end{figure}

\item Among all the network topologies, there are 14 topologies which yield the highest gains over both random Gaussian coding and interference avoidance. As an example, one of these networks is depicted in Figure \ref{fig12}.
\begin{figure}
\centering
\begin{subfigure}[b]{0.48\textwidth}
\centering
\includegraphics[trim = 2in 3in 2in 3.5in, clip,width=0.6\textwidth]{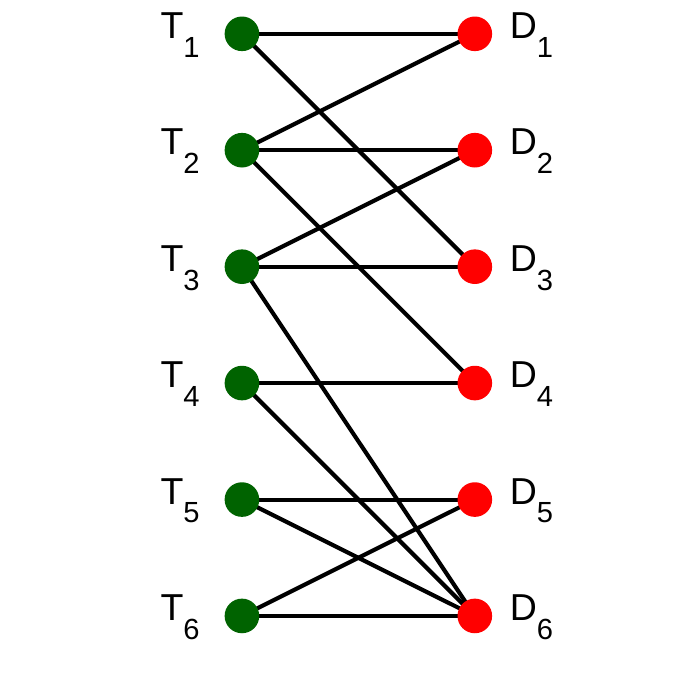}
\caption{}
\label{fig12a}
\end{subfigure}
~
\begin{subfigure}[b]{0.48\textwidth}
\centering
\includegraphics[trim = 1.2in 3in .9in 2.5in, clip,width=0.8\textwidth]{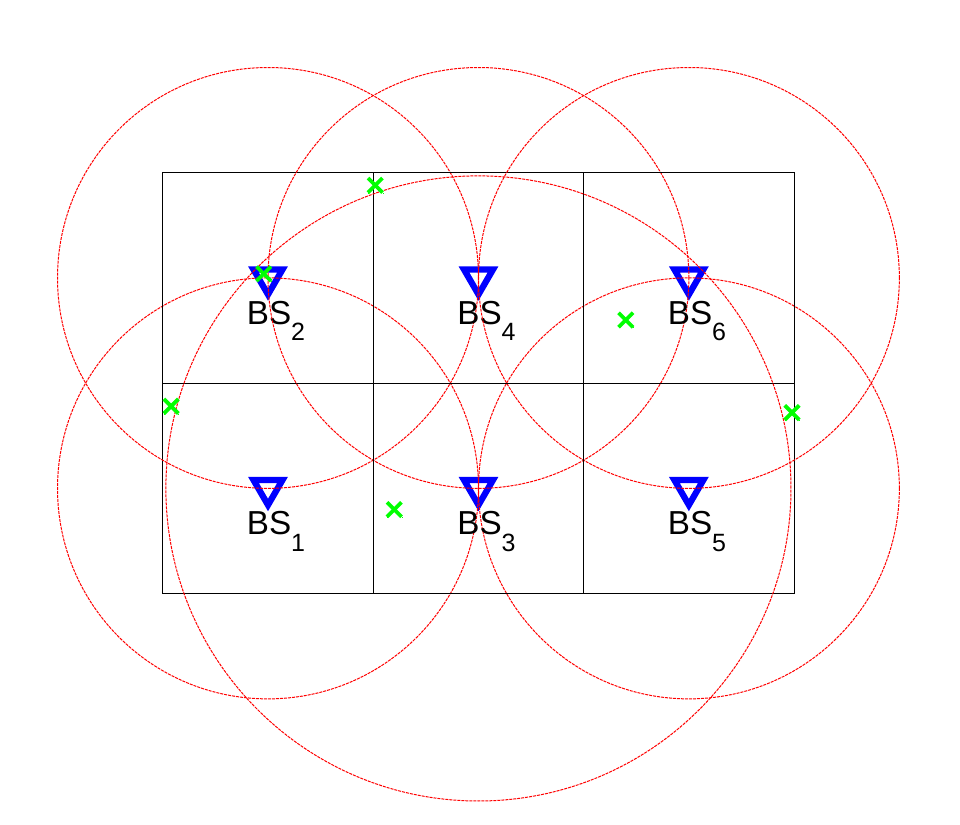}
\caption{}
\label{fig12b}
\end{subfigure}
\caption{(a) A 6-user interference network in which $d_{sym}=\frac{1}{2}$ and the gain of structured repetition coding over random Gaussian coding and interference avoidance is 2 and $\frac{3}{2}$, respectively, and (b) a corresponding 6-cell realization.}
\label{fig12}
\end{figure}

In the network of Figure \ref{fig12} (and all the other 13 networks which yield the highest gains), $d_{sym}$ is equal to $\frac{1}{2}$, which can be achieved by structured repetition coding. However, the best symmetric DoF achieved by random Gaussian coding is $\frac{1}{4}$, hence a gain of 2 can be obtained over this scheme. This implies that for all these 14 networks, there exists a receiver whose degree is 4 (receiver $\text{D}_6$ in Figure \ref{fig12a}). Moreover, another pattern that is common among these 14 ``high-yield'' topologies is that the three users which are interfering to the receiver with degree 4 are mutually non-interfering, hence constituting an independent set (users \{3,4,5\} in Figure \ref{fig12a}). The third common property of all these topologies is that they contain a 3-user cyclic chain (a 3-user network with users $i$,$j$ and $k$ where $\text{T}_i$ is connected to $\text{D}_j$, $\text{T}_j$ is connected to $\text{D}_k$, and $\text{T}_k$ is connected to $\text{D}_i$). The subgraph consisting of users \{1,2,3\} in Figure \ref{fig12a} is a 3-user cyclic chain. This is the main reason that interference avoidance can achieve no better than the symmetric DoF of $\frac{1}{3}$ in these networks, allowing structured repetition coding to have a gain of $\frac{3}{2}$ over it.
\end{enumerate}

\subsection{6-User Networks with 1 Central and 5 Surrounding Base Stations}\label{7MBS}

In this section, we explore another class of 6-user networks, consisting of 1 base station (BS) located in the center of a circle with radius 1, and 5 other base stations located uniformly on the boundary of the circle. Each base station has a coverage radius of $r<1$, with a mobile client randomly located in its coverage area. A realization of such a network scenario is illustrated in Figure \ref{fig15}.
\begin{figure}[hbt]
\centering
\includegraphics[trim = 2in 3.1in 1.65in 2.95in, clip,width=0.36\textwidth]{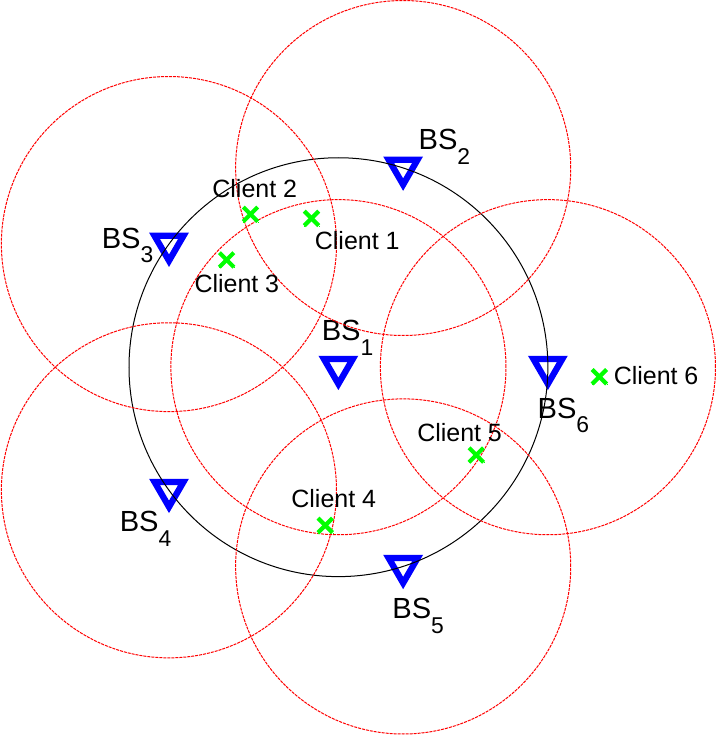}
\caption{A 6-user network realization with 1 BS in the middle and 5 BS's surrounding it, where the blue triangles, green crosses, black circle and red circles represent base stations, mobile clients, unit circle and coverage area of base stations, respectively. In this figure, $r=0.8$.}
\label{fig15}
\end{figure}

Again, as we had in Figure \ref{fig8}, the blue triangles represent base stations, the green crosses represent mobile clients, the black circle represents the unit circle and the red circles depict the coverage area of each BS. Obviously, any placement of the mobile clients corresponds to a partially-connected 6-user interference network.

To analyze our bounds for this class of networks, we generated 12000 network instances by randomly locating the mobile clients for the case of $r=0.8$. Upon removing isomorphic graphs, we ended up with 1507 distinct topologies and evaluated our inner and outer bounds for these topologies, leading to the following conclusions.

\begin{enumerate}
\item We find out interestingly, that our bounds are tight in all the generated network topologies, and Figure \ref{fig16} illustrates the distribution of $d_{sym}$ among these topologies. We note that $d_{sym}$ takes 4 distinct values in $\lbrace\frac{1}{3},\frac{2}{5},\frac{1}{2},1\rbrace$. The most frequent value that $d_{sym}$ takes is $\frac{1}{3}$, followed by $\frac{1}{2}$.

\begin{figure}[hbt]
\centering
\includegraphics[trim = 1.65in 3.3in 1.75in 3.4in, clip,width=.35\textwidth]{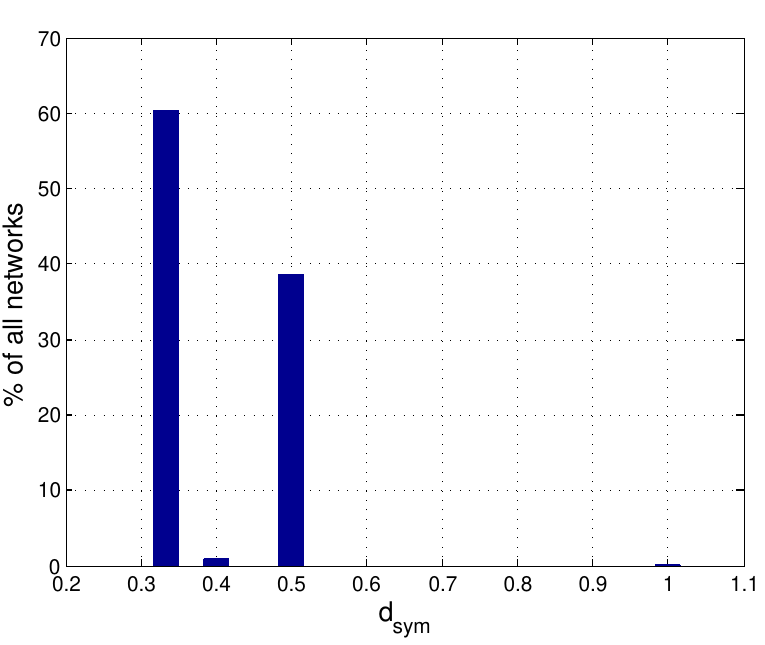}
\caption{Distribution of $d_{sym}$ among 6-user networks with 1 central and 5 surrounding BS's, where each BS has a coverage radius of $r=0.8$.}
\label{fig16}
\end{figure}

\item Figure \ref{fig_13} illustrates the effect of the number of cross links in the network, which is a measure of density of the network graph, on the fraction of topologies which yield gains over benchmark schemes. The trend is similar to that of Figure \ref{fig10}, showing that if the network graph is too sparse (few number of cross links) or too dense (high number of cross links), there is not much gain beyond the benchmark schemes. However, if the network graph is moderately dense, then structured repetition coding can attain gain over the benchmark schemes in a larger fraction of networks. Moreover, the figure implies that, on average, interference avoidance yields higher inner bounds on $d_{sym}$ than random Gaussian coding, in this class of networks.
\begin{figure}[h]
\centering
\includegraphics[trim = 0in 2.5in 0in 2.6in, clip,width=0.52\textwidth]{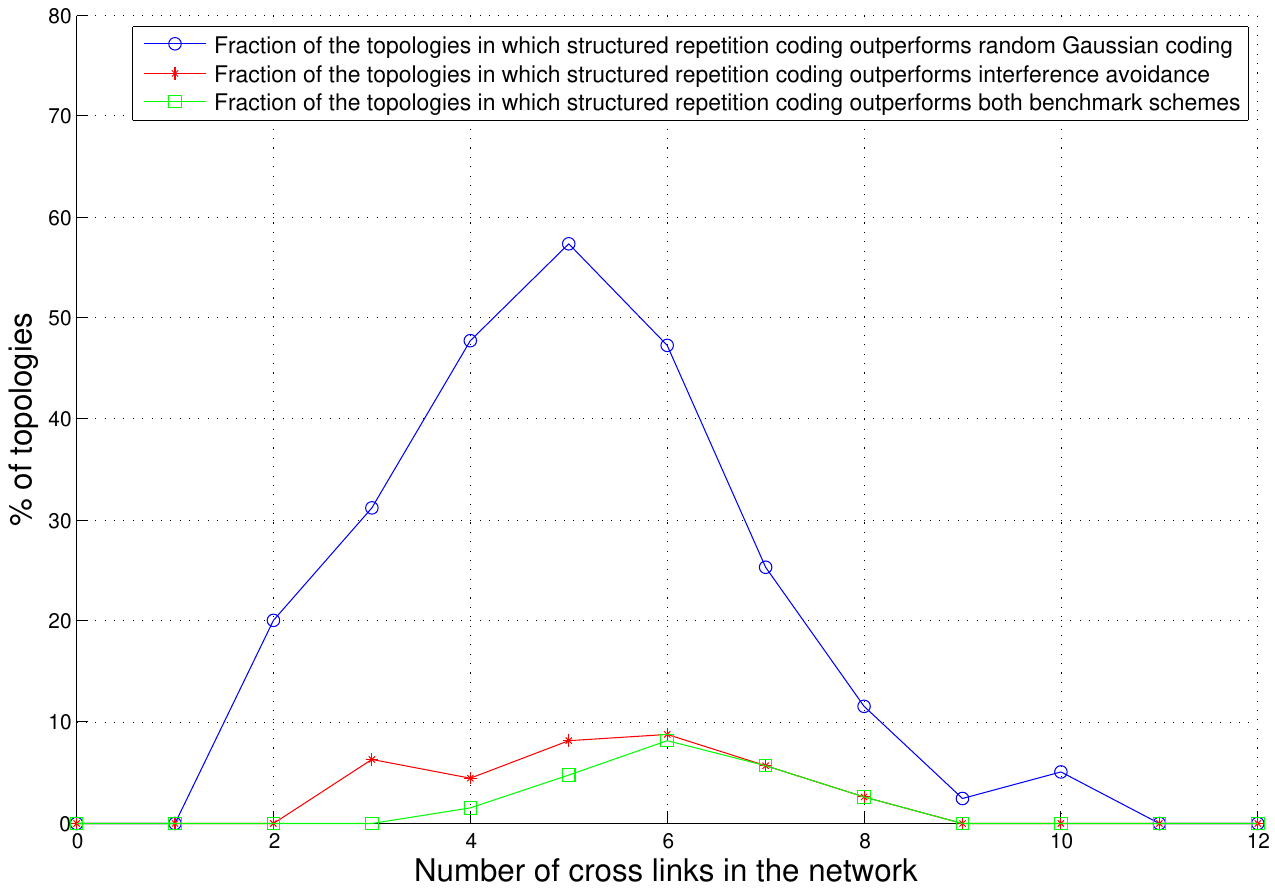}
\caption{Effect of network density on the fraction of networks in which structured repetition coding outperforms benchmark schemes in 6-user networks with 1 central and 5 surrounding BS's, where each BS has a coverage radius of $r=0.8$.}
\label{fig_13}
\end{figure}

\item Figure \ref{fig18a} illustrates the distribution of the gain of structured repetition coding over random Gaussian coding among the topologies in which this gain is greater than unity. This gain can take 2 distinct values in $\lbrace \frac{6}{5},\frac{3}{2}\rbrace$. Moreover, among the networks in which there is a gain over interference avoidance, this gain can take 2 distinct values in $\lbrace \frac{5}{4},\frac{3}{2}\rbrace$, with the distribution shown in Figure \ref{fig18b}. The most frequent value of both of the gains is $\frac{3}{2}$, which indicates a 50\% improvement in the inner bound on $d_{sym}$.
\begin{figure}[hbt]
\centering
\begin{subfigure}[b]{0.47\textwidth}
\centering
\includegraphics[trim = 1.6in 3.3in 1.75in 3.4in, clip,width=0.8\textwidth]{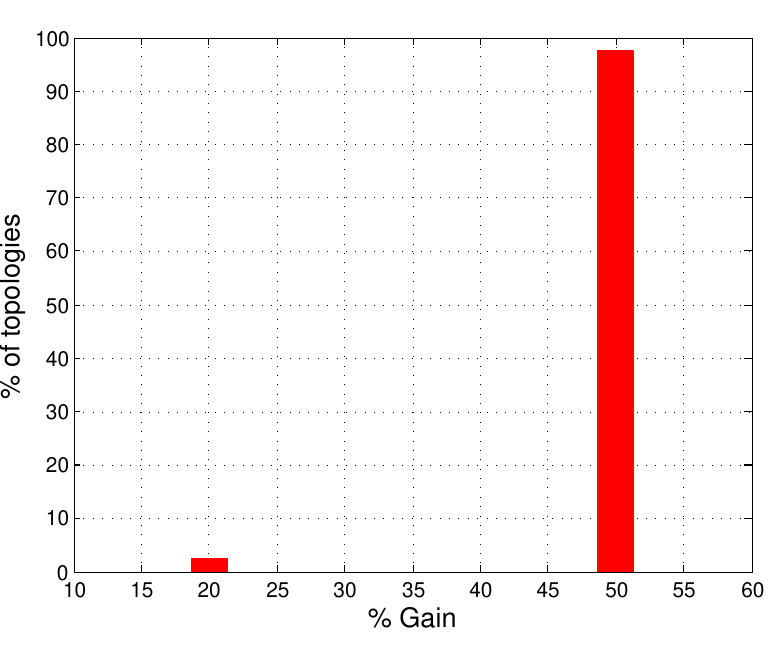}
\caption{Distribution of the gain of structured repetition coding over $\frac{1}{\Delta_{R}}$ (random Gaussian coding).}
\label{fig18a}
\end{subfigure}
~
\begin{subfigure}[b]{0.47\textwidth}
\centering
\includegraphics[trim = 1.6in 3.3in 1.75in 3.4in, clip,width=0.8\textwidth]{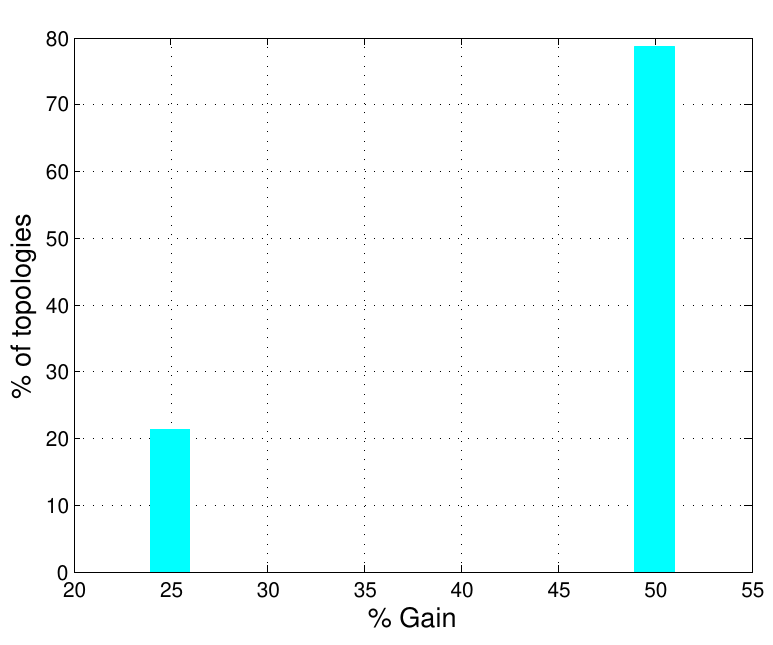}
\caption{Distribution of the gain of structured repetition coding over interference avoidance.}
\label{fig18b}
\end{subfigure}
\caption{Comparison of achievable schemes in 6-user networks with 1 central and 5 surrounding BS's, where each BS has a coverage radius of $r=0.8$.}
\label{fig18}
\end{figure}
\end{enumerate}

\section{Concluding Remarks and Future Directions}\label{conc}
In this work, we studied the impact of network topology on the symmetric degrees-of-freedom of $K$-user interference networks with no CSIT. We presented two outer bounds on the symmetric DoF based on two new linear algebraic concepts of generators and fractional generators. An achievable scheme, called structured repetition coding, has been introduced based on the graph theoretic concept of bipartite matching, as well as two benchmark achievable schemes. Moreover, we demonstrated, via numerical analysis, that our bounds were tight for most topologies in two classes of networks. We illustrated topologies in which structured repetition coding yields gains over benchmark schemes, and also discussed the effect of network sparsity on these gains.

This paper also opens up several interesting future directions. For instance, as we demonstrated in Section \ref{sim}, our bounds were tight for most instances of network topologies. However, we found some networks in which our bounds were not tight. Figure \ref{fig_conc} illustrates two of these networks.
\begin{figure}[hbt]
\centering
\begin{subfigure}[hbt]{0.45\textwidth}
\centering
\includegraphics[trim = 2in 3in 2in 3.5in, clip,width=0.7\textwidth]{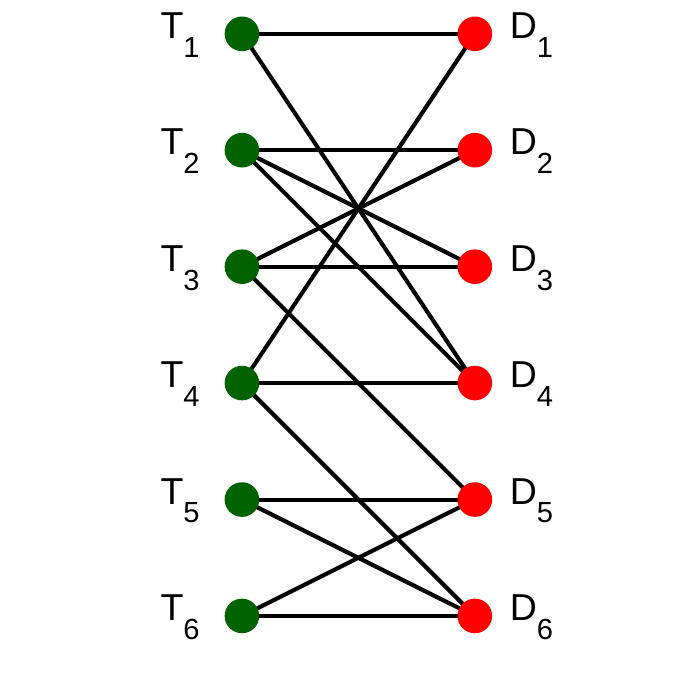}
\caption{}
\label{fig_conc_a}
\end{subfigure}
~
\begin{subfigure}[hbt]{0.45\textwidth}
\centering
\includegraphics[trim = 2in 3in 2in 3.5in, clip,width=0.7\textwidth]{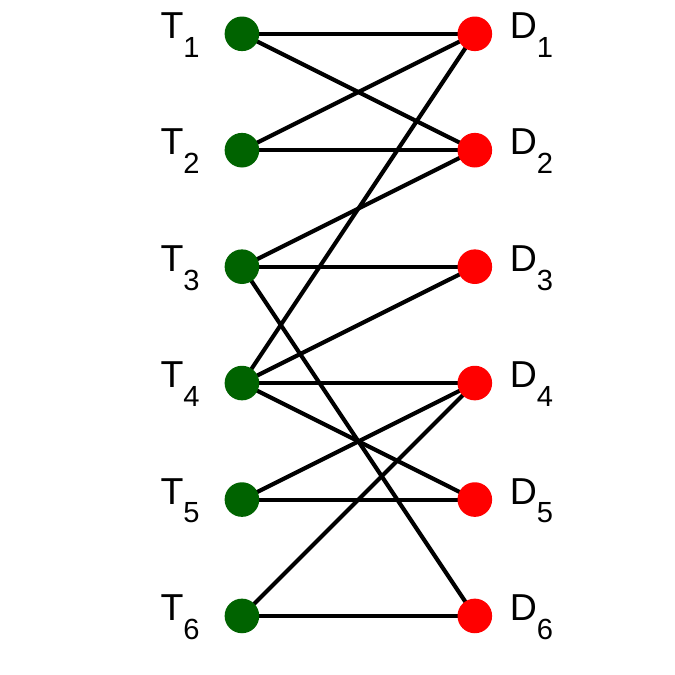}
\caption{}
\label{fig_conc_b}
\end{subfigure}
\caption{Two interference networks in which our bounds yield $\frac{4}{9}\leq d_{sym}\leq\frac{1}{2}$, hence characterizing the symmetric DoF remains open.}
\label{fig_conc}
\end{figure}

Both of these networks correspond to the class of 6-user networks with 6 square cells, discussed in Section \ref{6cell}. By imposing the constraint $n\in[1:K+1]$ on the transmission matrices of structured repetition coding, our best inner bounds in the networks of Figures \ref{fig_conc_a} and \ref{fig_conc_b} are  $\frac{2}{5}$ and $\frac{1}{3}$, respectively. However, letting $n=9$ leads to the inner bound of $\frac{4}{9}$. This shows that in some topologies, letting $n>K+1$ may yield higher inner bounds by structured repetition coding than $n\leq K+1$. However, our inner and outer bounds still do not meet in these networks. Therefore, an interesting direction would be finding new techniques to tighten the bounds for these networks and extending them to general topologies.

Another interesting direction is the generalization of the problem to more generic network settings, such as multihop networks. As an example, the authors in \cite{kkk} study the problem of two-hop wireless networks with $K$ sources, $K$ relays and $K$ destinations, where the network is assumed to be fully-connected and full CSI is also presumed to be available at the transmitters. This problem can be extended in two ways by considering the impacts of partial connectivity and lack of CSIT on the results.

\section*{Acknowledgement}
The authors would like to thank Dr. Bhushan Naga at Qualcomm Inc. for his valuable comments and discussions, and Dr. Syed Ali Jafar for his comments regarding the examples in Section \ref{conc}, pointing out that the symmetric DoF of $\frac{5}{12}$ can also be achieved, which shows that letting $n>K+1$ improves the achievable symmetric DoF under structured repetition coding.

\appendices
\section{Proof of Lemma \ref{lem2}}\label{apx1}
\begin{align*}
H(W|Y^n+Z_2^n)
&=H(W|Y^n+Z_2^n,Z_1^n-Z_2^n)+I(W;Z_1^n-Z_2^n|Y^n+Z_2^n)\\
&\leq H(W|Y^n+Z_1^n)+h(Z_1^n-Z_2^n|Y^n+Z_2^n)-h(Z_1^n-Z_2^n|Y^n+Z_2^n,W)\\
&\leq n\epsilon+h(Z_1^n-Z_2^n)-h(Z_1^n-Z_2^n|Y^n+Z_2^n,W,Z_2^n)\\
&= n\epsilon+h(Z_1^n-Z_2^n)-h(Z_1^n)\\
&=n\epsilon+n\log(\pi e (N+1))-n\log(\pi e)\\
&=n\epsilon+n\log(N+1).
\end{align*}

\section{Proof of Lemma \ref{lem3}}\label{apx2}

Without loss of generality, let $\mathcal{S}=[1:m]$ and $\mathcal{S'}=[1:m']$ ($m'\leq m$). Also, with respect to $\mathbf{c}$ being a fractional generator of $\mathcal{S'}$ in $\mathcal{S}$, suppose (without loss of generality) that $\Pi_\mathcal{S'}=(1,...,m')$. First, note that
\begin{align*}
&h\left(\sum_{j=1}^m \mathbf{c}_j g_j^n X_j^n+\sum_{k=1}^{m'}  g_k^n X_k^n+Z^n|\mathcal{G}^n]\right)
-h\left(\sum_{j=1}^m \mathbf{c}_j g_j^n X_j^n+\sum_{k=1}^{m'}  g_k^n X_k^n+Z^n|W_1,...,W_{m'},\mathcal{G}^n\right)\\
&\qquad=H(W_1,...,W_{m'})-H\left(W_1,...,W_{m'}|\sum_{j=1}^m \mathbf{c}_j g_j^n X_j^n+\sum_{k=1}^{m'}  g_k^n X_k^n+Z^n,\mathcal{G}^n\right),
\end{align*}
since both sides are equal to $I\left(\sum_{j=1}^m \mathbf{c}_j g_j^n X_j^n+\sum_{k=1}^{m'}  g_k^n X_k^n+Z^n;W_1,...,W_{m'}|\mathcal{G}^n\right)$. Therefore, we can write
\begin{align*}
&h\left(\sum_{j=1}^m \mathbf{c}_j g_j^n X_j^n+\sum_{k=1}^{m'}  g_k^n X_k^n+Z^n|W_1,...,W_{m'},\mathcal{G}^n\right)\\
&\qquad=H\left(W_1,...,W_{m'}|\sum_{j=1}^m \mathbf{c}_j g_j^n X_j^n+\sum_{k=1}^{m'}  g_k^n X_k^n+Z^n,\mathcal{G}^n\right)\\
&\qquad\qquad+h\left(\sum_{j=1}^m \mathbf{c}_j g_j^n X_j^n+\sum_{k=1}^{m'}  g_k^n X_k^n+Z^n|\mathcal{G}^n)-H(W_1,...,W_{m'}\right)\\
&\qquad\leq H\left(W_1,...,W_{m'}|\sum_{j=1}^m \mathbf{c}_j g_j^n X_j^n+\sum_{k=1}^{m'}  g_k^n X_k^n+Z^n,\mathcal{G}^n\right)
+n\left(\log(P)- \sum_{i\in\mathcal{S'}} R_i\right)+no(\log(P)).\numberthis\label{eq11}
\end{align*}

Now, we prove that $H\left(W_l|\sum_{j=1}^m \mathbf{c}_j g_j^n X_j^n+\sum_{k=1}^{m'}  g_k^n X_k^n+Z^n,W_1,...,W_{l-1},\mathcal{G}^n\right)\leq no(\log(P))+n\epsilon_{l,n}$ for $l\in[1:m']$. By Definition \ref{d2}, $\mathbf{M}_{l}^\mathcal{S}\in_{l}^\pm \text{span}\	(\mathbf{c}+\sum_{k=1}^{m'} \mathbf{I}_k^{|\mathcal{S}|},\mathbf{I}_{\{1,...,{l-1}\}}^{|\mathcal{S}|})$, implying that there exist a vector $\tilde{\mathbf{v}}\in\mathbb{R}^{|\mathcal{S}|}$ and coefficients $\alpha$ and $d_k$ ($k\in[1:l-1]$) such that
\begin{align}
\tilde{\mathbf{v}}&=\alpha \left(\mathbf{c}+\sum_{k=1}^{m'} \mathbf{I}_k^{|\mathcal{S}|}\right)+\sum_{k=1}^{l-1} d_k \mathbf{I}_k^{|\mathcal{S}|}\label{eqqq1}\\
|\tilde{\mathbf{v}}_l|&=|\mathbf{M}_{ll}^\mathcal{S}|=1\label{eqqq2}\\
\tilde{\mathbf{v}}_j\left(|\tilde{\mathbf{v}}_j|-|\mathbf{M}_{jl}^\mathcal{S}|\right)&=0,\:\forall j\in[1:m]\setminus\{l\}.\label{eqqq3}
\end{align}

Note that if $j\in\mathcal{IF}_l$, then $\mathbf{M}_{jl}^\mathcal{S}=1$ and (\ref{eqqq3}) implies that $\tilde{\mathbf{v}}_j$ can either be equal to 0 or $\pm1$; i.e. $\tilde{\mathbf{v}}_j\in\{0,\pm1\}$. On the other hand, if $j\notin\mathcal{IF}_l$, then $\mathbf{M}_{jl}^\mathcal{S}=0$ and (\ref{eqqq3}) implies that $\tilde{\mathbf{v}}_j=0$. Multiplying $\begin{bmatrix} g_1^n X_1^n & \dots & g_m^n X_m^n \end{bmatrix}$ by both sides of (\ref{eqqq1}), hence, yields
\begin{align*}
\tilde{\mathbf{v}}_l g_l^n X_l^n+\sum_{j\in\mathcal{IF}_l}\tilde{\mathbf{v}}_j g_j^n X_j^n=\alpha \bigg(\sum_{j=1}^m \mathbf{c}_j g_j^n X_j^n+\sum_{k=1}^{m'}  g_k^n X_k^n \bigg) +\sum_{k=1}^{l-1} d_{k} g_k^n X_{k}^n.
\end{align*}

Therefore, we can write:
\begin{align*}
H\Bigg(W_l|\alpha &\bigg(\sum_{j=1}^m \mathbf{c}_j g_j^n X_j^n+\sum_{k=1}^{m'}  g_k^n X_k^n+Z^n \bigg) +\sum_{k=1}^{l-1} d_{k} g_k^n X_{k}^n,\mathcal{G}^n\Bigg)\\
&=H\left(W_l|\tilde{\mathbf{v}}_l g_l^n X_l^n+\sum_{j\in\mathcal{IF}_l}\tilde{\mathbf{v}}_j g_j^n X_j^n+\alpha{Z}^n,\mathcal{G}^n\right)\\
&=H\left(W_l|\tilde{\mathbf{v}}_l g_l^n X_l^n+\sum_{j\in\mathcal{IF}_l}\tilde{\mathbf{v}}_j g_j^n X_j^n+\alpha{Z}^n,\sum_{j\in\mathcal{IF}_l}(1-|\tilde{\mathbf{v}}_j|) g_j^n X_j^n,\mathcal{G}^n\right)\numberthis\label{eqq6}\\
&\leq H\left(W_l|\tilde{\mathbf{v}}_l g_l^n X_l^n+\sum_{j\in\mathcal{IF}_l}\tilde{\mathbf{v}}'_j g_j^n X_j^n+\alpha{Z}^n,\mathcal{G}^n\right)\numberthis\label{eqq6p}\\
&\leq no(\log(P))+n\epsilon_{l,n},\numberthis\label{eqq7}
\end{align*}
where (\ref{eqq6}) is true because, as discussed before, for all $j\in\mathcal{IF}_l$, $\tilde{\mathbf{v}}_j$ can only take the values in $\{\pm1,0\}$ and therefore the signals in $\sum_{j\in\mathcal{IF}_l}\tilde{\mathbf{v}}_j g_j^n X_j^n$ and $\sum_{j\in\mathcal{IF}_l}(1-|\tilde{\mathbf{v}}_j|) g_j^n X_j^n$ do not have common terms.\footnote{If $\tilde{\mathbf{v}}_j=0$, then $1-|\tilde{\mathbf{v}}_j|=1$, and if $\tilde{\mathbf{v}}_j=1$ or $\tilde{\mathbf{v}}_j=-1$, then $1-|\tilde{\mathbf{v}}_j|=0$. Hence, either $\tilde{\mathbf{v}}_j$ or $1-|\tilde{\mathbf{v}}_j|$ is non-zero, but not both.} In (\ref{eqq6p}), $\tilde{\mathbf{v}}'_j$ is defined as $\tilde{\mathbf{v}}'_j:=\tilde{\mathbf{v}}_j+(1-|\tilde{\mathbf{v}}_j|)$. Clearly $\tilde{\mathbf{v}}'_j$ can only take the values in $\{+1,-1\}$ because $\tilde{\mathbf{v}}_j\in\{\pm1,0\}$. Also, (\ref{eqqq2}) implies that $\tilde{\mathbf{v}}_l\in\{+1,-1\}$. Therefore, $\tilde{\mathbf{v}}_l g_l^n X_l^n+\sum_{j\in\mathcal{IF}_l}\tilde{\mathbf{v}}'_j g_j^n X_j^n+\alpha{Z}^n$ is statistically the same as $Y_l^n$ (with a bounded difference in noise variance), because the channel gains have a symmetric distribution around zero ($f_G(g)=f_G(-g)$, $\forall g\in\mathbb{C}$). This, together with Lemma \ref{lem2} and Fano's inequality, implies that (\ref{eqq7}) is correct. Hence, using the chain rule for differential entropy yields
\begin{align*}
H \bigg(W_1,...,W_{m'}|\sum_{j=1}^m \mathbf{c}_j g_j^n X_j^n&+\sum_{k=1}^{m'}  g_k^n X_k^n+Z^n,\mathcal{G}^n \bigg)\\
&=\sum_{l=1}^{m'} H\bigg(W_l|\sum_{j=1}^m \mathbf{c}_j g_j^n X_j^n+\sum_{k=1}^{m'}  g_k^n X_k^n+Z^n,W_1,...,W_{l-1},\mathcal{G}^n\bigg)\\
&\leq \sum_{l=1}^{m'} no(\log(P))+n\epsilon_{l,n}\\
&= no(\log(P))+n\epsilon_{n}.
\end{align*}

Therefore, (\ref{eq11}) can be written as
\begin{align}\label{eq13}
h\bigg(\sum_{j=1}^m \mathbf{c}_j g_j^n X_j^n+\sum_{k=1}^{m'}  g_k^n X_k^n+Z^n|W_1,...,W_{m'},\mathcal{G}^n\bigg)
\leq n\bigg(\log(P)- \sum_{i\in\mathcal{S'}} R_i\bigg)+no(\log(P))+n\epsilon_{n}.
\end{align}

But note that
\begin{align}\label{eq14}
h\left(\sum_{j=1}^m \mathbf{c}_j g_j^n X_j^n+\sum_{k=1}^{m'}  g_k^n X_k^n+Z^n|W_1,...,W_{m'},\mathcal{G}^n\right)
=h\left(\sum_{j=1}^m \mathbf{c}_j g_j^n X_j^n+Z^n|\mathcal{G}^n\right),
\end{align}
because by Definition \ref{d2}, $\mathbf{c}_j=0,\forall j\in\mathcal{S}'$ and therefore, $\sum_{j=1}^m \mathbf{c}_j g_j^n X_j^n+Z^n$ is independent of $W_1,...,W_{m'}$. The lemma then follows from (\ref{eq13}) and (\ref{eq14}).

\section{Proof of Lemma \ref{match}}\label{apx3}
The fact that $\mu(G)-\mu(G\setminus l)=1$ means that there exists a maximum matching in $G$ which covers node $v_l$; i.e. there is one edge in the matching incident on $v_l$. This matching covers $\mu(G)$ vertices out of $\{v_1,...,v_m\}$, including $v_l$ for sure, and $\mu(G)$ vertices out of $\{v'_1,...,v'_n\}$ (note that for the bipartite graph $G$, $\mu(G)\leq\min\{m,n\}$). Therefore, it corresponds to a $\mu(G)\times\mu(G)$ submatrix of the entire adjacency matrix $\mathbf{T}$, which we will denote by $\mathbf{T}^{l}$, and we know that $\mathbf{T}^{l}$ includes (a subset of) the $l^{th}$ row of $\mathbf{T}$. Without loss of generality, we assume that $l\in[1:\mu(G)]$ and $\mathbf{T}^{l}$ consists of the first $\mu(G)$ rows and columns of $\mathbf{T}$. We will also denote the corresponding random matrix by $\tilde{\mathbf{T}}^l$; i.e. $\tilde{\mathbf{T}}_{ij}^l=g_{ij}\mathbf{T}_{ij}^l$, $\forall i\in[1:\mu(G)]$, $\forall j\in[1:\mu(G)]$.

Now, we show that $\det(\tilde{\mathbf{T}}^l)\neq0$ for almost all values of $g_{ij}$'s. This is because 
\begin{align}\label{eq_det}
\det(\tilde{\mathbf{T}}^l)=\sum_{\sigma\in\Pi_{\mu(G)}}\text{sgn}(\sigma)\prod_{i=1}^{\mu(G)}\tilde{\mathbf{T}}^l_{i\sigma_i},
\end{align}
where $\Pi_{\mu(G)}$ is the set of all permutations of $[1:\mu(G)]$ and $\text{sgn}(\sigma)=1$ if $\sigma$ can be derived from $[1:\mu(G)]$ by doing an even number of switches, and $\text{sgn}(\sigma)=-1$ otherwise. Note that $\det(\tilde{\mathbf{T}}^l)$ is a multivariate polynomial of distinct i.i.d. channel gains (drawn from a continuous distribution), which is not identically zero. The reason that the polynomial is not identically zero is because the matching corresponds to a set of nonzero entries of $\mathbf{T}^l$ (and hence $\tilde{\mathbf{T}}^l$) which do not share common rows/columns, hence constituting a non-zero term in (\ref{eq_det}). Therefore, the Schwartz-Zippel Lemma \cite{sch,zip} states that the value of this polynomial is not zero for almost all values of $g_{ij}$'s.

Therefore, $\tilde{\mathbf{T}}^l$ is invertible with probability 1, implying that there exists a vector $\mathbf{u}'\in\mathbb{C}^{\mu(G)}$ such that $\tilde{\mathbf{T}}^l\mathbf{u}'=\mathbf{I}_l^{\mu(G)}$. In fact, $\mathbf{u}'=(\tilde{\mathbf{T}}^l)^{-1}\mathbf{I}_l^{\mu(G)}$.
Now, let $\mathbf{u}\in\mathbb{C}^n$ be the vector such that
\begin{align}\label{uup}
\mathbf{u}_j=
\begin{cases}
\mathbf{u}'_j&\:\text{if }j\in[1:\mu(G)]\\
0&\:\text{if }j\in[\mu(G)+1:n]
\end{cases}.
\end{align}

Now, we claim that $\tilde{\mathbf{T}}\mathbf{u}=\mathbf{I}_l^m$. This is true because of the following. As a reminder, the subscript * refers to the corresponding row of a matrix.
\begin{itemize}
\item $\tilde{\mathbf{T}}_{l,*}\mathbf{u}
=\tilde{\mathbf{T}}_{l,*}^l\mathbf{u}'=1$.

\item $\tilde{\mathbf{T}}_{i,*}\mathbf{u}
=\tilde{\mathbf{T}}_{i,*}^l\mathbf{u}'=0$, $\forall i\in[1:\mu(G)]\setminus\{l\}$.

\item Also, each row $\tilde{\mathbf{T}}_{j,*}$ ($j\in[\mu(G)+1:n]$) is linearly dependent on the rows $\tilde{\mathbf{T}}_{i,*}$, $i\in[1:\mu(G)]\setminus\{l\}$, because otherwise, we would have at least $\mu(G)$ independent rows in $\tilde{\mathbf{T}}_{\setminus l,*}$ (the same matrix as $\tilde{\mathbf{T}}$ with the $l^{th}$ row removed) and this corresponds to a matching with a size of at least $\mu(G)$ in $G\setminus l$, contradicting $\mu(G)-1$ being the size of the maximum matching in $G\setminus l$. Therefore, for all $j\in[\mu(G)+1:n]$, there exist coefficients $\alpha_{ij}$ ($i\in[1:\mu(G)]\setminus\{l\}$) such that $\tilde{\mathbf{T}}_{j,*}=\sum_{i\in[1:\mu(G)]\setminus\{l\}} \alpha_{ij} \tilde{\mathbf{T}}_{i,*}$, implying that $\tilde{\mathbf{T}}_{j,*}\mathbf{u}=\sum_{i\in[1:\mu(G)]\setminus\{l\}} \alpha_{ij} \tilde{\mathbf{T}}_{i,*}\mathbf{u}=0$.
\end{itemize}

To complete the proof, note that $\|\mathbf{u}\|_2=\|\mathbf{u}'\|_2=\left\|(\tilde{\mathbf{T}}^l)^{-1}\mathbf{I}_l^{\mu(G)}\right\|_2$, because of the definition of $\mathbf{u}$ in (\ref{uup}).

\section{Proof of Finiteness of Noise Variance in (\ref{eq_noise})}\label{apx4}

In this appendix, we intend to show that in (\ref{eq_noise}), $\mathbb{E}\left[\log\left(\left\|\left(\tilde{\mathbf{G}}^{j,l}\right)^{-1} \mathbf{I}_l^{\mu(\bar{G}^j)}\right\|^2_2\right)\right]<\infty$. We need the following key lemma to prove this inequality.

\begin{lem}\footnote{This lemma has connections to estimating the size of lemniscates of multivariate polynomials, studied in \cite{math1,math2}. However, here we prove a different form of upper bound which suits our framework to prove the finiteness of the noise variance in (\ref{eq_noise}).}\label{poly}
Assume $p(X_1,...,X_n)=\sum_{i=1}^m a_i \prod_{j=1}^n X_j^{d_{ji}}$ ($1\leq m\leq 2^n$) is a multivariate polynomial of complex i.i.d. random variables $X_1,...,X_n$ with a continuous distribution, where for all $i\in[1:m]$, $a_i$ is a constant coefficient in $\mathbb{C}$ satisfying $|a_i|\geq 1$, all the monomials are assumed to be distinct, and the degree of $X_j$ in the $i^{th}$ monomial, denoted by $d_{ji}$, satisfies $d_{ji}\in\{0,1\}$, $\forall j\in[1:n],\forall i\in[1:m]$. If there exists $f_{max}<\infty$ such that $f_{|X|}(r)\leq f_{max},\forall r\in\mathbb{R}^+$, where $f_{|X|}(r)$ is the distribution of $|X_i|$, $\forall i\in[1:n]$, then for all $\epsilon\in[0,1]$,
\begin{align*}
\emph{Pr}\left[|p(X_1,...,X_n)|\leq\epsilon\right]\leq 2^{n+1}f_{max} {\epsilon}^{\frac{1}{2^{n-1}}}.
\end{align*}
\end{lem}

Before proving the lemma, we show how this lemma implies that $\mathbb{E}\left[\log\left(\left\|\left(\tilde{\mathbf{G}}^{j,l}\right)^{-1} \mathbf{I}_l^{\mu(\bar{G}^j)}\right\|^2_2\right)\right]<\infty$ in (\ref{eq_noise}). Note that $\left(\tilde{\mathbf{G}}^{j,l}\right)^{-1} \mathbf{I}_l^{\mu(\bar{G}^j)}$ is the $l^{th}$ column of $\left(\tilde{\mathbf{G}}^{j,l}\right)^{-1}$, the inverse of $\tilde{\mathbf{G}}^{j,l}$. Therefore,
\begin{align*}
\left\|\left(\tilde{\mathbf{G}}^{j,l}\right)^{-1} \mathbf{I}_l^{\mu(\bar{G}^j)}\right\|^2_2
&=\sum_{i=1}^{\mu(\bar{G}^j)}\left|\left(\tilde{\mathbf{G}}^{j,l}\right)_{il}^{-1} \right|^2\\
&=\frac{1}{\left|\det\left(\tilde{\mathbf{G}}^{j,l}\right)\right|^2}\sum_{i=1}^{\mu(\bar{G}^j)}\left|M_{li}\right|^2,\numberthis\label{inv}
\end{align*}
where in (\ref{inv}), we have replaced $\left(\tilde{\mathbf{G}}^{j,l}\right)_{il}^{-1}=\frac{M_{li}}{\det\left(\tilde{\mathbf{G}}^{j,l}\right)}$, $M_{li}$ being the determinant of $\tilde{\mathbf{G}}^{j,l}$ after removing its $l^{th}$ row and $i^{th}$ column. Due to the definition of determinants (see (\ref{eq_det}), for instance), each of the terms $M_{li}$ and also $\det\left(\tilde{\mathbf{G}}^{j,l}\right)$ is a multivariate polynomial of i.i.d. channel gains, in which each of the random variables appear with the degree of 0 or 1 in each monomial. In other words, if we rename the i.i.d. channel gains inside $\tilde{\mathbf{G}}^{j,l}$ as $g_1,...,g_n$, then each $M_{li}$ can be written as
\begin{align}\label{mli}
M_{li}=\sum_{k=1}^{m_{li}} a_{k,li}\prod_{h=1}^n g_h^{d_{k,h,li}},
\end{align}
and $\det\left(\tilde{\mathbf{G}}^{j,l}\right)$ can be written as
\begin{align}\label{detG}
\det\left(\tilde{\mathbf{G}}^{j,l}\right)=\sum_{k=1}^{m} a_{k}\prod_{h=1}^n g_h^{d_{k,h}},
\end{align}
where $|a_{k,li}|=|a_k|=1$ and $d_{k,h,li},d_{k,h}\in\{0,1\}$, for all $h,k,i$. Hence, we can write
\begin{align*}
\mathbb{E}\left[\log\left(\left\|\left(\tilde{\mathbf{G}}^{j,l}\right)^{-1} \mathbf{I}_l^{\mu(\bar{G}^j)}\right\|^2_2\right)\right]
&=\mathbb{E}\left[\log\left(\sum_{i=1}^{\mu(\bar{G}^j)}\left|M_{li}\right|^2\right)\right]-\mathbb{E}\left[\log \left|\det\left(\tilde{\mathbf{G}}^{j,l}\right)\right|^2\right]\\
&\leq \log\left(\sum_{i=1}^{\mu(\bar{G}^j)} \mathbb{E}\left[\left|M_{li}\right|^2\right]\right)-\mathbb{E}\left[\log \left|\det\left(\tilde{\mathbf{G}}^{j,l}\right)\right|^2\right]\numberthis\label{jensen}\\
&\leq \log\left(\sum_{i=1}^{\mu(\bar{G}^j)}m_{li}\sum_{k=1}^{m_{li}} \mathbb{E}\left[\left|\prod_{h=1}^n g_h^{d_{k,h,li}}\right|^2\right]\right)-\mathbb{E}\left[\log \left|\det\left(\tilde{\mathbf{G}}^{j,l}\right)\right|^2\right]\numberthis\label{cauchy}\\
&= \log\left(\sum_{i=1}^{\mu(\bar{G}^j)}m_{li}\sum_{k=1}^{m_{li}} \mathbb{E}\left[\left| g\right|^2\right]^{\sum_{h=1}^n d_{k,h,li}}\right)-\mathbb{E}\left[\log \left|\det\left(\tilde{\mathbf{G}}^{j,l}\right)\right|^2\right]\numberthis\label{indep},
\end{align*}
where (\ref{jensen}) follows from Jensen's inequality, (\ref{cauchy}) follows from the Cauchy-Schwarz inequality, and (\ref{indep}) follows from $g_h$'s being i.i.d. Hence, if $\mathbb{E}\left[\left| g\right|^2\right]<\infty$, then the first term in (\ref{indep}) is bounded. Therefore, it remains to show that $\mathbb{E}\left[\log \left|\det\left(\tilde{\mathbf{G}}^{j,l}\right)\right|^2\right]>-\infty$. We can write
\begin{align}\label{detY}
\mathbb{E}\left[\log \left|\det\left(\tilde{\mathbf{G}}^{j,l}\right)\right|^2\right]\geq \mathbb{E}\left[\log \left(\min\left\{1,\left|\det\left(\tilde{\mathbf{G}}^{j,l}\right)\right|^2\right\}\right)\right]=-\mathbb{E}[Y],
\end{align}
where $Y=-\log \left(\min\left\{1,\left|\det\left(\tilde{\mathbf{G}}^{j,l}\right)\right|^2\right\}\right)$ is a non-negative random variable. Hence, we have
\begin{align*}
\mathbb{E}[Y]&=\int_0^\infty \text{Pr}[Y\geq y]dy\\
&=\int_0^\infty \text{Pr}\left[\min\left\{1,\left|\det\left(\tilde{\mathbf{G}}^{j,l}\right)\right|^2\right\}\leq 2^{-y}\right] dy\\
&=\int_0^\infty \text{Pr}\left[\left|\det\left(\tilde{\mathbf{G}}^{j,l}\right)\right|^2\leq 2^{-y}\right] dy\numberthis\label{mindrop}\\
&=\frac{2}{\ln 2}\int_0^1 \text{Pr}\left[\left|\det\left(\tilde{\mathbf{G}}^{j,l}\right)\right|\leq u\right] \frac{du}{u}\numberthis\label{varchange}\\
&\leq \frac{2}{\ln 2}\int_0^1  \frac{2^{n+1}f_{max} u^{\frac{1}{2^{n-1}}} }{u} du\numberthis\label{lemuse}\\
&= \frac{2^{n+2}f_{max}}{\ln 2}\int_0^1  u^{\frac{1}{2^{n-1}}-1} du\\
&=\frac{2^{2n+1}f_{max}}{\ln2}\\
&<\infty,
\end{align*}
where (\ref{mindrop}) is true because for any random variable $X$ and any constant $c< 1$, total probability law implies
\begin{align*}
\text{Pr}\left[\min\{1,X\}\leq c\right]&=\text{Pr}\left[\min\{1,X\}\leq c|X\leq 1\right] \text{Pr}[X\leq 1]+\text{Pr}\left[\min\{1,X\}\leq c|X> 1\right] \text{Pr}[X> 1]\\
&=\text{Pr}\left[X\leq c|X\leq 1\right] \text{Pr}[X\leq 1]+\cancelto{0}{\text{Pr}\left[1\leq c\right]}~\quad \text{Pr}[X> 1]\\
&=\text{Pr}\left[X\leq c~,X\leq 1\right]\\
&=\text{Pr}[X\leq c].
\end{align*}

Moreover, in (\ref{varchange}) we have used the change of variables $u=2^{-\frac{y}{2}}$ and (\ref{lemuse}) follows from (\ref{detG}) and Lemma \ref{poly}. This, together with (\ref{detY}) implies that $\mathbb{E}\left[\log \left|\det\left(\tilde{\mathbf{G}}^{j,l}\right)\right|^2\right]>-\infty$, hence finishing the proof.

Now, we focus on proving Lemma \ref{poly}.

\begin{proof}[Proof of Lemma \ref{poly}]
We will use induction on the number of variables ($n$) to prove the desired inequality.

\textbf{Base case:} We need to prove that for all $\epsilon\leq 1$, $\text{Pr}\left[|p(X_1)|\leq\epsilon\right]\leq 4f_{max} {\epsilon}$. In general, $p(X_1)=aX_1+b$, where $a,b\in\mathbb{C}$ and $|a|\geq 1$ and $|b|\geq 1$. Therefore, we can write
\begin{align*}
\text{Pr}\left[|p(X_1)|\leq\epsilon\right]&=\text{Pr}\left[|aX_1+b|\leq\epsilon\right]\\
&\leq \text{Pr}\left[\left| |aX_1|-|b|\right|\leq\epsilon\right]\\
&= \text{Pr}\left[\frac{|b|-\epsilon}{|a|}\leq |X_1|\leq \frac{|b|+\epsilon}{|a|}\right]\\
&\leq{2 f_{max}\epsilon}\numberthis\label{baseq}\\
&< 4f_{max}\epsilon.
\end{align*}

\textbf{Inductive step:} Assume for all $\epsilon\in[0,1]$, $\text{Pr}\left[|p(X_1,...,X_{k})|\leq\epsilon\right]\leq 2^{k+1}f_{max} {\epsilon}^{\frac{1}{2^{k-1}}}$. Now, consider the polynomial $p(X_1,...,X_k,X_{k+1})=\sum_{i=1}^m a_i \prod_{j=1}^{k+1} X_j^{d_{ji}}$. Without loss of generality, we can write this polynomial as
\begin{align}
p(X_1,...,X_k,X_{k+1})=\sum_{i=1}^{m'}(a_i X_{k+1}+b_i)  \prod_{j=1}^{k} X_j^{d_{ji}}+\sum_{i=m'+1}^{m} a_i \prod_{j=1}^{k} X_j^{d_{ji}},
\end{align}
where we first factored out the monomials which include $X_{k+1}$, and afterwards, we lumped together the monomials that were indistinct in terms of $X_1,...,X_k$. 

Now, we can write
\begin{align*}
&\text{Pr}\left[|p(X_1,...,X_k,X_{k+1})|\leq\epsilon\right]\\
&=\text{Pr}\left[|p(X_1,...,X_k,X_{k+1})| \leq\epsilon\middle\vert\underset{i\in[1:m']}{\min}|a_iX_{k+1}+b_i|\leq\sqrt{\epsilon}\right]\text{Pr}\left[\underset{i\in[1:m']}{\min}|a_iX_{k+1}+b_i|\leq\sqrt{\epsilon}\right]\\
&\qquad+\text{Pr}\left[|p(X_1,...,X_k,X_{k+1})| \leq\epsilon\middle\vert\underset{i\in[1:m']}{\min}|a_iX_{k+1}+b_i|>\sqrt{\epsilon}\right]\text{Pr}\left[\underset{i\in[1:m']}{\min}|a_iX_{k+1}+b_i|>\sqrt{\epsilon}\right]\\
&\leq \text{Pr}\left[\underset{i\in[1:m']}{\min}|a_iX_{k+1}+b_i|\leq\sqrt{\epsilon}\right]
+\iint\limits_A
\text{Pr}\left[|p(X_1,...,X_k,re^{j\phi})| \leq\epsilon\right]
f_{|X|,\angle X}(r,\phi) d\phi dr\numberthis\label{intA}\\
&\leq \sum_{i=1}^{m'}\text{Pr}[|a_iX_{k+1}+b_i|\leq\sqrt{\epsilon}]\\
&\qquad+\iint\limits_A
\text{Pr}\left[\left|\sum_{i=1}^{m'}\frac{a_i re^{j\phi}+b_i}{\sqrt{\epsilon}}  \prod_{j=1}^{k} X_j^{d_{ji}}+\sum_{i=m'+1}^{m} \frac{a_i}{\sqrt{\epsilon}} \prod_{j=1}^{k} X_j^{d_{ji}}\right| \leq\sqrt{\epsilon}\right]
f_{|X|,\angle X}(r,\phi) d\phi dr\numberthis\label{eq_base}\\
&\leq 2^k (2 f_{max}\sqrt{\epsilon})
+\iint\limits_A
\left(2^{k+1}f_{max} {\sqrt{\epsilon}}^{\frac{1}{2^{k-1}}}\right)
f_{|X|,\angle X}(r,\phi) d\phi dr\numberthis\label{eq_lem}\\
&\leq 2^{k+1}f_{max} {\epsilon}^{\frac{1}{2^{k}}}+2^{k+1}f_{max} {\epsilon}^{\frac{1}{2^{k}}}\\
&=2^{k+2}f_{max}{\epsilon}^{\frac{1}{2^{k}}},
\end{align*}
where in (\ref{intA}-\ref{eq_lem}), the integration is over $A=\{(r,\phi):\underset{i\in[1:m']}{\min}|a_i re^{j\phi}+b_i|>\sqrt{\epsilon}\}$, and in (\ref{eq_base}), we have used the union bound. Also, (\ref{eq_lem}) is true because of the upper bound in (\ref{baseq}), the fact that $m'\leq 2^k$, and also because in (\ref{eq_base}), we have $\left|\frac{a_ire^{j\phi}+b_i}{\sqrt{\epsilon}}\right|> 1,\forall i\in[1:m']$ and $\left|\frac{a_i}{\sqrt{\epsilon}}\right|\geq |a_i|\geq 1,\forall i\in[m'+1:m]$, which enables us to use the inductive assumption by noting that $\sum_{i=1}^{m'}\frac{a_i re^{j\phi}+b_i}{\sqrt{\epsilon}}  \prod_{j=1}^{k} X_j^{d_{ji}}+\sum_{i=m'+1}^{m} \frac{a_i}{\sqrt{\epsilon}} \prod_{j=1}^{k} X_j^{d_{ji}}$ is a polynomial in $X_1,...,X_k$, satisfying the conditions in the lemma. This completes the proof.
\end{proof}


\begin{thebibliography}{9}

\bibitem{isit}
N. Naderializadeh and A. S. Avestimehr, ``Impact of Topology on Interference Networks with No CSIT,'' in  \emph{Proceedings of IEEE International Symposium on Information Theory}, Istanbul, Turkey, 2013.

\bibitem{localview}
V. Aggarwal, A. S. Avestimehr, and A. Sabharwal, ``On Achieving Local View Capacity Via Maximal Independent Graph Scheduling,'' \emph{IEEE Transactions on Information Theory}, vol. 57, no. 5, pp. 2711-2729, May 2011.











\bibitem{jafar_old}
S. A. Jafar, ``Elements of Cellular Blind Interference Alignment --- Aligned Frequency Reuse, Wireless Index Coding and Interference Diversity,'' available online at \emph{arXiv:1203.2384}.

\bibitem{jafar}
S. A. Jafar, ``Topological Interference Management through Index Coding,'' \emph{IEEE Transactions on Information Theory}, vol. 60, no. 1, pp. 529-568, January 2014.

\bibitem{zchannel}
Y. Zhu and D. Guo, ``On the Capacity Region of Fading Z-Interference Channels without CSIT,'' in \emph{Proceedings of IEEE International Symposium on Information Theory}, Austin, TX, USA, 2010.

\bibitem{huang}
C. Huang, S. A. Jafar, S. Shamai, and S. Vishwanath, ``On Degrees of Freedom Region of MIMO Networks without CSIT,'' available online at \emph{arXiv:0909.4017}.

\bibitem{vaze_no_csit}
C. S. Vaze and M. K. Varanasi, ``The Degree-of-Freedom Regions of MIMO Broadcast, Interference, and Cognitive Radio Channels With No CSIT,'' \emph{IEEE Transactions on Information Theory}, vol. 58, no. 8, pp. 5354-5374, August 2012.

\bibitem{sina}
S. Lashgari, A. S. Avestimehr, and C. Suh, ``Linear Degrees of Freedom of the X-Channel with Delayed CSIT,'' \emph{IEEE Transactions on Information Theory}, vol. 60, no. 4, pp. 2180-2189, April 2014.

\bibitem{ali}
A. Vahid, M. A. Maddah-Ali, and A. S. Avestimehr, ``Capacity Results for Binary Fading Interference Channels with Delayed CSIT,'' \emph{IEEE Transactions on Information Theory}, vol. 60, no. 10, pp. 6093-6130, October 2014.

\bibitem{maleki}
H. Maleki, S. A. Jafar, and S. Shamai, ``Retrospective Interference Alignment,'' in  \emph{Proceedings of IEEE International Symposium on Information Theory}, Saint-Petersburg, Russia, 2011.

\bibitem{abdoli}
M. J. Abdoli, A. Ghasemi, and A. K. Khandani, ``On the Degrees of Freedom of $K$-User SISO Interference and X Channels with Delayed CSIT,''  available online at \emph{arXiv:1109.4314}.

\bibitem{vaze}
C. S. Vaze and M. K. Varanasi, ``The Degrees of Freedom Region and Interference Alignment for the MIMO Interference Channel With Delayed CSIT,'' \emph{IEEE Transactions on Information Theory}, vol. 58, no. 7, pp. 4396-4417, July 2012.

\bibitem{vvv}
A. El Gamal, V. S. Annapureddy, and V. V. Veeravalli, ``Degrees of Freedom (DoF) of Locally Connected Interference Channels with Cooperating Multiple-Antenna Transmitters,'' in \emph{Proceedings of IEEE International Symposium on Information Theory}, Cambridge, MA, USA, 2012.

\bibitem{wang}
C. Wang, S. A. Jafar, S. Shamai, and M. Wigger, ``Interference, Cooperation and Connectivity — A Degrees of Freedom Perspective,'' in  \emph{Proceedings of IEEE International Symposium on Information Theory}, Saint-Petersburg, Russia, 2011.

\bibitem{gastpar}
S. W. Jeon, N. Goela, and M. Gastpar, ``Degrees of Freedom of Sparsely Connected Wireless Networks,'' in \emph{Proceedings of IEEE International Symposium on Information Theory}, Cambridge, MA, USA, 2012.

\bibitem{hong}
L. Ruan and V. K. N. Lau, ``Dynamic Interference Mitigation for Generalized Partially Connected Quasi-Static MIMO Interference Channel,'' \emph{IEEE Transactions on Signal Processing}, vol. 59, no. 8, pp. 3788-3798, August 2011.

\bibitem{tse}
D. Tse and P. Viswanath, \emph{Fundamentals of Wireless Communication}, Cambridge University Press, 2005.

\bibitem{fgt}
E. R. Scheinerman and D. H. Ullman, \emph{Fractional Graph Theory: A Rational Approach to the Theory of Graphs}, John Wiley \& Sons, 2008.

\bibitem{kkk}
I. Shomorony and A. S. Avestimehr, ``Degrees of Freedom of Two-Hop Wireless Networks: ``Everyone Gets the Entire Cake'','' \emph{IEEE Transactions on Information Theory}, vol. 60, no. 5, pp. 2417-2431, May 2014.

\bibitem{sch}
J. T. Schwartz, ``Fast Probabilistic Algorithms for Verification of Polynomial Identities,'' \emph{Journal of the ACM}, 27(4):701-717, 1980.

\bibitem{zip}
R. Zippel, ``Probabilistic Algorithms for Sparse Polynomials,'' in \emph{Proceedings of the International Symposium on Symbolic and Algebraic Computation}, pp. 216-226, 1979.


\bibitem{math1}
A. Cuyt, K. Driver, and D. S. Lubinsky, ``On the Size of Lemniscates of Polynomials in One and Several Variables,'' in \emph{Proceedings of the American Mathematical Society}, vol. 124, no. 7, pp. 2123-2136, July 1996.

\bibitem{math2}
D. S. Lubinsky, ``Small Values of Polynomials: Cartan, P\'{o}lya and Others,'' \emph{Journal of Inequalities and Applications}, vol. 1, pp. 199-222, 1997.

\end{thebibliography}
\end{document}